\documentclass{amsart}
\usepackage{amsmath,amssymb,amsthm}
\usepackage{tkz-graph}

\newcommand{\re}{\mathbb{R}}
\newcommand{\bc}{\mathbb{C}}
\newcommand{\bh}{\mathbb{H}}
\newcommand{\zz}{\mathbb{Z}}
\newcommand{\zt}{\zz_2{}}
\newcommand{\gf}{\mathbb{F}}
\DeclareMathOperator{\Cl}{Cl}
\DeclareMathOperator{\wt}{wt}
\newtheorem{thm}{Theorem}

\newtheorem{proposition}{Proposition}

\begin{document}
\title[Adinkras: Graphs of Clifford Algebra Representations]{Adinkras: Graphs of Clifford Algebra Representations, Supersymmetry, and Codes}
\author[K. Iga]{Kevin Iga}
\date{September 2021}
\address{
Pepperdine University\\
Natural Science Division\\
Malibu CA 90263-4321\\
USA}
\email{kiga@pepperdine.edu}
\maketitle

\begin{abstract}
An Adinkra is a graph from the study of supersymmetry in particle physics, but it can be adapted to study Clifford algebra representations.  The graph in this context is called a Cliffordinkra, and puts some standard ideas in Clifford algebra representations in a geometric and visual context.

In the past few years there have been developments in Adinkras that have shown how they are connected to error correcting codes, algebraic topology, algebraic geometry, and combinatorics.  These connections also arise for Cliffordinkras.

This paper introduces Cliffordinkras and describes the relationship to these subjects in that context.  No previous knowledge of Adinkras and supersymmetry is assumed.
\end{abstract}

\section{Introduction}
An {\em Adinkra} is a directed graph, together with a coloring of edges, with some edge being dashed and others solid, satisfying various conditions which will be described later.  A {\em Cliffordinkra} is analogous, except the edges are undirected.  See Figure~\ref{fig:adinkraex}.

\begin{figure}[h]
\begin{center}
\begin{tikzpicture}[scale=0.15]
\GraphInit[vstyle=Welsh] %Allow coloring of vertices
\SetVertexSimple[MinSize=5pt]  %No labeling of vertices, set size
\SetUpEdge[labelstyle={draw},style={ultra thick,->}]
\tikzset{Dash/.style = {ultra thick, ->, dashed,draw}}
\Vertex[x=0,y=5]{A}
\Vertex[x=0,y=10]{B}
\Vertex[x=0,y=20]{C}
\Vertex[x=10,y=10]{D}
\Vertex[x=-5,y=10]{E}
\Vertex[x=5,y=10]{F}
\Vertex[x=-5,y=20]{G}
\Vertex[x=5,y=20]{H}
\AddVertexColor{white}{A,C,G,H}
\Edge[color=red](B)(G)
\Edge[color=red](D)(H)
\Edge[color=red](F)(C)
\Edge[color=red](A)(E)
\Edge[color=blue](D)(G)
\Edge[color=blue, style=Dash](B)(H)
\Edge[color=blue](E)(C)
\Edge[color=blue, style=Dash](A)(F)
\Edge(D)(C)
\Edge[style=Dash](F)(H)
\Edge[style=Dash](E)(G)
\Edge(A)(B)
\end{tikzpicture}
\makebox[.5in]{}
\begin{tikzpicture}[scale=0.15]
\GraphInit[vstyle=Welsh] %Allow coloring of vertices
\SetVertexSimple[MinSize=5pt]  %No labeling of vertices, set size
\SetUpEdge[labelstyle={draw},style={ultra thick}]
\tikzset{Dash/.style = {ultra thick, dashed,draw}}
\Vertex[x=0,y=5]{A}
\Vertex[x=0,y=10]{B}
\Vertex[x=0,y=20]{C}
\Vertex[x=10,y=10]{D}
\Vertex[x=-5,y=10]{E}
\Vertex[x=5,y=10]{F}
\Vertex[x=-5,y=20]{G}
\Vertex[x=5,y=20]{H}
\AddVertexColor{white}{A,C,G,H}
\Edge[color=red](B)(G)
\Edge[color=red](D)(H)
\Edge[color=red](F)(C)
\Edge[color=red](A)(E)
\Edge[color=blue](D)(G)
\Edge[color=blue, style=Dash](B)(H)
\Edge[color=blue](E)(C)
\Edge[color=blue, style=Dash](A)(F)
\Edge(D)(C)
\Edge[style=Dash](F)(H)
\Edge[style=Dash](E)(G)
\Edge(A)(B)
\end{tikzpicture}
\end{center}
\caption{An example of an Adinkra (left) and a Cliffordinkra (right).
\label{fig:adinkraex}}
\end{figure}
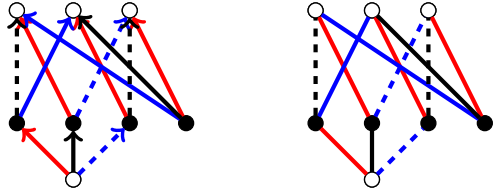

Adinkras were first introduced to understand supersymmetry, a concept in particle physics, and as such, record information on certain kinds of representations of the supersymmetry algebra.\cite{rA}  A Cliffordinkra likewise records information about the representation of a Clifford algebra.\cite{Doran:2008rp,r6--1}  Even though Adinkras were developed first, Cliffordinkras are a bit simpler, and the relationship to the Clifford algebra is more natural and straightforward.

Cliffordinkras provide a new graph-theoretic perspective on representations of Clifford algebras.  They allow for a new kind of geometric intuition in understanding Clifford algebras, and highlight the combinatorial structure.  Along the way, this idea leads to a classification of sets of signed permutation matrices $\Gamma_1,\ldots,\Gamma_n$ that satisfy the Clifford algebra relations.  In addition, it provides a framework for constructing concrete sets of such matrices.

It turns out that the mathematics behind Adinkras goes well beyond graph theory and includes error correcting codes,\cite{r6-codes,at0,at,r6-8} cubical cohomology,\cite{cc} and Riemann surfaces.\cite{geom1,geom2}  These relationships to other subjects do not depend on the orientation of the edges, so the same is true for Cliffordinkras.

This paper reports on these developments in the context of Cliffordinkras and Clifford algebras.  No previous familiarity with Adinkras or supersymmetry is assumed, and indeed, the main point of this paper is that these concepts can be introduced in the Clifford algebra context more naturally than in the original supersymmetry context.

\section{Cliffordinkras}
\label{sec:cliffordinkra}
A {\bf Cliffordinkra} with $n$ colors is an undirected graph with the following decorations:
\begin{itemize}
\item Bipartite: Every vertex is either a boson (drawn with an open circle) or a fermion (drawn with a filled circle), and every edge connects a boson with a fermion.
\item Edge colors: Every edge is colored one of $n$ colors.
\item Dashing: Every edge is either drawn as solid or dashed.
\end{itemize}
These decorations satisfy the following conditions:
\begin{itemize}
\item Regular edge-colored: Every vertex is incident to exactly one edge of each color.
\item Quadrilateral: Given any two different colors, the subgraph consisting of edges of those two colors is a disjoint union of cycles of length 4.  Each such cycle involving two colors is called a bicolor cycle, and this requirement says that bicolor cycles are all length 4.
\item Totally odd dashing: Every bicolor cycle has an odd number of dashed edges.
\end{itemize}

An example of a Cliffordinkra with 4 colors is in Figure~\ref{fig:cliffordinkraex}.  We are using the colors black, red, green, and blue.  The bosons are E, F, G, and H, and the fermions are A, B, C, and D.

\begin{figure}[h]
\begin{center}
\begin{tikzpicture}[scale=0.15]
\GraphInit[vstyle=Welsh] %Allow coloring of vertices
\SetVertexNormal[MinSize=5pt]  %No labeling of vertices, set size
\SetUpEdge[labelstyle={draw},style={ultra thick}]
\tikzset{Dash/.style = {ultra thick,dashed,draw}}
\Vertex[x=-10,y=0,Lpos=270,L={E}]{G}
\Vertex[x=0,y=0,Lpos=270,L={F}]{A}
\Vertex[x=10,y=0,Lpos=270,L={G}]{H}
\Vertex[x=20,y=0,Lpos=270,L={H}]{C}
\Vertex[x=-10,y=10,Lpos=90,L={A}]{E}
\Vertex[x=0,y=10,Lpos=90]{B}
\Vertex[x=10,y=10,Lpos=90,L={C}]{F}
\Vertex[x=20,y=10,Lpos=90]{D}
%\AddVertexColor{white}{A,C,G,H}
\AddVertexColor{black}{B,F,D,E}
\Edge[color=red](B)(G)
\Edge[color=red](D)(H)
\Edge[color=red](F)(C)
\Edge[color=red](A)(E)
\Edge[color=blue](D)(G)
\Edge[color=blue, style=Dash](B)(H)
\Edge[color=blue](E)(C)
\Edge[color=blue, style=Dash](A)(F)
\Edge(D)(C)
\Edge[style=Dash](F)(H)
\Edge[style=Dash](E)(G)
\Edge(A)(B)
\Edge[color=green](G)(F)
\Edge[color=green,style=Dash](H)(E)
\Edge[color=green](A)(D)
\Edge[color=green,style=Dash](C)(B)
\end{tikzpicture}
\end{center}
\caption{An example of a Cliffordinkra with 4 colors.
\label{fig:cliffordinkraex}}
\end{figure}
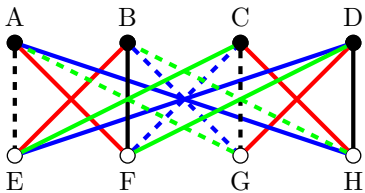

Vertex A is incident to a black dashed edge that connects to E, a red solid edge going to F, a green dashed edge going to G, and a blue solid edge going to H.  These are the four colors, each used exactly once.  Likewise there are four edges incident to B, one of each color, and so on for each vertex in the graph.  Thus, this graph is {\em regular edge-colored}.

If we take the subgraph with only the black and red edges, the result is Figure~\ref{fig:bicolorex}.  This is a disjoint union of two cycles: one with vertices A, F, B, E, in that cyclic order; the other cycle with vertices C, H, D, G.  These are both cycles of length 4, which demonstrates the {\em quadrilateral} property.  They both have 3 solid edges and one dashed edge: since the number of dashed edges is odd, this demonstrates the {\em totally odd dashing} property.  The same statements are true for any pair of colors.  For instance, if we choose the colors blue and green, the subgraph is shown in Figure~\ref{fig:bicolorex2}, and the reader can verify that these same properties hold in that case, as well.

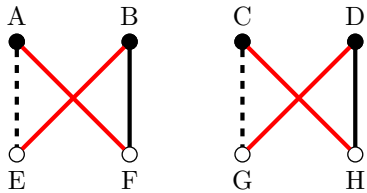
\begin{figure}[h]
\begin{center}
\begin{tikzpicture}[scale=0.15]
\GraphInit[vstyle=Welsh] %Allow coloring of vertices
\SetVertexNormal[MinSize=5pt]  %No labeling of vertices, set size
\SetUpEdge[labelstyle={draw},style={ultra thick}]
\tikzset{Dash/.style = {ultra thick,dashed,draw}}
\Vertex[x=-10,y=0,Lpos=270,L={E}]{G}
\Vertex[x=0,y=0,Lpos=270,L={F}]{A}
\Vertex[x=10,y=0,Lpos=270,L={G}]{H}
\Vertex[x=20,y=0,Lpos=270,L={H}]{C}
\Vertex[x=-10,y=10,Lpos=90,L={A}]{E}
\Vertex[x=0,y=10,Lpos=90]{B}
\Vertex[x=10,y=10,Lpos=90,L={C}]{F}
\Vertex[x=20,y=10,Lpos=90]{D}
%\AddVertexColor{white}{A,C,G,H}
\AddVertexColor{black}{B,F,D,E}
\Edge[color=red](B)(G)
\Edge[color=red](D)(H)
\Edge[color=red](F)(C)
\Edge[color=red](A)(E)
\Edge(D)(C)
\Edge[style=Dash](F)(H)
\Edge[style=Dash](E)(G)
\Edge(A)(B)
\end{tikzpicture}
\end{center}
\caption{The subgraph consisting of the black and red edges.
\label{fig:bicolorex}}
\end{figure}

\begin{figure}[h]
\begin{center}
\begin{tikzpicture}[scale=0.15]
\GraphInit[vstyle=Welsh] %Allow coloring of vertices
\SetVertexNormal[MinSize=5pt]  %No labeling of vertices, set size
\SetUpEdge[labelstyle={draw},style={ultra thick}]
\tikzset{Dash/.style = {ultra thick,dashed,draw}}
\Vertex[x=-10,y=0,Lpos=270,L={E}]{G}
\Vertex[x=0,y=0,Lpos=270,L={F}]{A}
\Vertex[x=10,y=0,Lpos=270,L={G}]{H}
\Vertex[x=20,y=0,Lpos=270,L={H}]{C}
\Vertex[x=-10,y=10,Lpos=90,L={A}]{E}
\Vertex[x=0,y=10,Lpos=90]{B}
\Vertex[x=10,y=10,Lpos=90,L={C}]{F}
\Vertex[x=20,y=10,Lpos=90]{D}
%\AddVertexColor{white}{A,C,G,H}
\AddVertexColor{black}{B,F,D,E}
\Edge[color=blue](D)(G)
\Edge[color=blue, style=Dash](B)(H)
\Edge[color=blue](E)(C)
\Edge[color=blue, style=Dash](A)(F)
\Edge[color=green](G)(F)
\Edge[color=green,style=Dash](H)(E)
\Edge[color=green](A)(D)
\Edge[color=green,style=Dash](C)(B)
\end{tikzpicture}
\end{center}
\caption{The subgraph consisting of the blue and green edges.
\label{fig:bicolorex2}}
\end{figure}
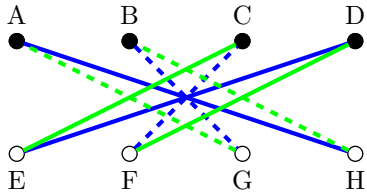

\subsection{Some Examples of Cliffordinkras}
\label{sec:cliffordinkraex}
We begin with a few examples of Cliffordinkras for $n=1$, $2$, and $3$.  For $n=1$, there is only one color, and every vertex must have only one edge incident to it.  Therefore, a Cliffordinkra with $n=1$ color is a disjoint union of line segments, as in the following example:

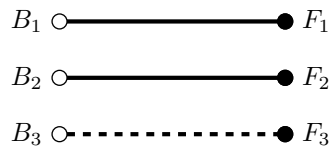
\begin{center}
\begin{tikzpicture}[scale=0.15]
\GraphInit[vstyle=Welsh]
\SetVertexNormal[MinSize=5pt]
\SetUpEdge[labelstyle={draw},style={ultra thick}]
\tikzset{Dash/.style={ultra thick,dashed,draw}}
\Vertex[x=0,y=0,Math,L={B_1},Math,Lpos=180]{B1}
\Vertex[x=20,y=0,Math,L={F_1},Math,Lpos=0]{F1}
\Vertex[x=0,y=-5,Math,L={B_2},Math,Lpos=180]{B2}
\Vertex[x=20,y=-5,Math,L={F_2},Math,Lpos=0]{F2}
\Vertex[x=0,y=-10,Math,L={B_3},Math,Lpos=180]{B3}
\Vertex[x=20,y=-10,Math,L={F_3},Math,Lpos=0]{F3}
\AddVertexColor{black}{F1,F2,F3}
\Edge(B1)(F1)
\Edge(B2)(F2)
\Edge[style=Dash](B3)(F3)
\end{tikzpicture}
\end{center}

More generally, the disjoint union of two Cliffordinkras with $n$ colors is a Cliffordinkra with $n$ colors, so it suffices to study connected Cliffordinkras.  There are two such connected Cliffordinkras with $n=1$ colors: a solid edge, like the one connecting $B_1$ and $F_1$ above, and a dashed edge, like the one connecting $B_3$ and $F_3$.

We now consider Cliffordinkras with $n=2$ colors (say, black and red).  Every vertex is incident to one black and one red edge.  By the Quadrilateral criterion in the definition of Cliffordinkras, we know that connected Cliffordinkras with 2 colors must be a cycle of order 4.  For instance:
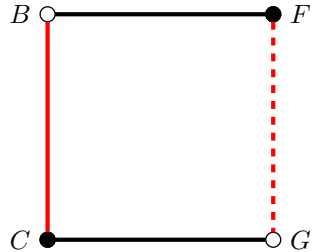
\begin{center}
\begin{tikzpicture}[scale=0.15]
\GraphInit[vstyle=Welsh]
\SetVertexNormal[MinSize=5pt]
\SetUpEdge[labelstyle={draw},style={ultra thick}]
\tikzset{Dash/.style={dashed,draw,ultra thick}}
\Vertex[x=0,y=0,Math,L={B},Lpos=180]{B}
\Vertex[x=20,y=0,Math,L={F},Lpos=0]{F}
\Vertex[x=0,y=-20,Math,L={C},Lpos=180]{C}
\Vertex[x=20,y=-20,Math,L={G},Lpos=0]{G}
\AddVertexColor{black}{F,C}
\Edge(B)(F)
\Edge(C)(G)
\Edge[color=red](B)(C)
\Edge[color=red,style=Dash](F)(G)
\end{tikzpicture}
\end{center}
Note the dashed red edge between $F$ and $G$.  By the Totally odd dashing criterion in the definition of Cliffordinkra, we know that this cycle must have either one or three dashed edges.

Below is an example of a Cliffordinkra with $n=3$ colors (black, red, and green).  This is a 3-dimensional cube, or more precisely its 1-skeleton (its vertices and edges).  Black edges are all horizontal, red edges are vertical, and blue edges go ``into the page''.

The bicolor cycles are the square faces of the cube.  The top face, for instance, is a blue-black bicolor cycle.  Each face indeed is a cycle of length 4, and each has an odd number of dashed edges.

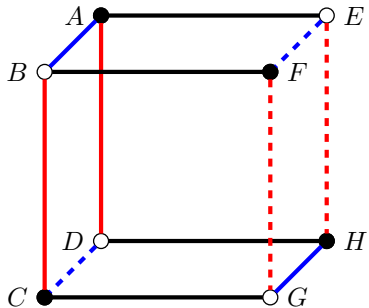
\begin{center}
\begin{tikzpicture}[scale=0.15]
\GraphInit[vstyle=Welsh]
\SetVertexNormal[MinSize=5pt]
\SetUpEdge[labelstyle={draw},style={ultra thick}]
\tikzset{Dash/.style={dashed,draw,ultra thick}}
\Vertex[x=0,y=0,Math,L={B},Lpos=180]{B}
\Vertex[x=20,y=0,Math,L={F},Lpos=0]{F}
\Vertex[x=0,y=-20,Math,L={C},Lpos=180]{C}
\Vertex[x=20,y=-20,Math,L={G},Lpos=0]{G}
\Vertex[x=5,y=5,Math,L={A},Lpos=180]{A}
\Vertex[x=25,y=5,Math,L={E},Lpos=0]{E}
\Vertex[x=5,y=-15,Math,L={D},Lpos=180]{D}
\Vertex[x=25,y=-15,Math,L={H},Lpos=0]{H}
\AddVertexColor{black}{F,C,A,H}
\Edge(A)(E)
\Edge(D)(H)
\Edge[color=red](A)(D)
\Edge[color=red,style=Dash](E)(H)
\Edge[color=blue](A)(B)
\Edge[color=blue,style=Dash](E)(F)
\Edge[color=blue,style=Dash](D)(C)
\Edge[color=blue](G)(H)
\Edge(B)(F)
\Edge(C)(G)
\Edge[color=red](B)(C)
\Edge[color=red,style=Dash](F)(G)
\end{tikzpicture}
\end{center}

In fact, every connected Cliffordinkra with $n=3$ colors is the $1$-skeleton of a $3$-dimensional cube, though there are many ways to assign a dashing that is totally odd.  This will be generalized in Section~\ref{sec:cubes}.

\subsection{Vertex switching}
\label{sec:vertexswitch}
Suppose we have a Cliffordinkra with $n$ colors, and let $V$ be the set of vertices and $E$ its set of edges.  Given a vertex $v\in V$ we can do a {\em vertex switch} on $v$ as follows.  For every edge incident to $v$, if it is solid, make it dashed.  If it is dashed, make it solid.  Figure~\ref{fig:vswitch} illustrates this.
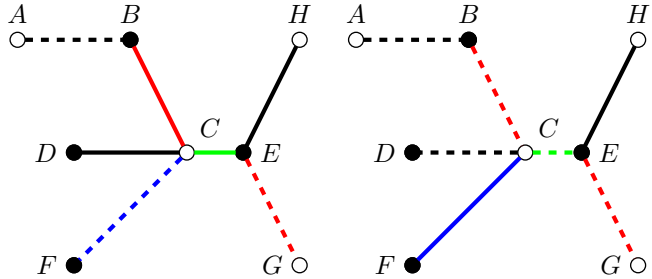
\begin{figure}[h]
\begin{center}
\begin{tikzpicture}[scale=0.15]
\GraphInit[vstyle=Welsh]
\SetVertexNormal[MinSize=5pt]
\SetUpEdge[labelstyle={draw},style={ultra thick}]
\tikzset{Dash/.style={dashed,draw,ultra thick}}
\Vertex[x=0,y=20,Math,L={A},Lpos=90]{A}
\Vertex[x=10,y=20,Math,L={B},Lpos=90]{B}
\Vertex[x=15,y=10,Math,L={C},Lpos=60]{C}
\Vertex[x=5,y=10,Math,L={D},Lpos=180]{D}
\Vertex[x=20,y=10,Math,L={E},Lpos=0]{E}
\Vertex[x=5,y=0,Math,L={F},Lpos=180]{F}
\Vertex[x=25,y=0,Math,L={G},Lpos=180]{G}
\Vertex[x=25,y=20,Math,L={H},Lpos=90]{H}
\AddVertexColor{black}{B,D,E,F}
\Edge[style=Dash](A)(B)
\Edge[color=red](B)(C)
\Edge(C)(D)
\Edge[color=green](C)(E)
\Edge[color=blue,style=Dash](C)(F)
\Edge[color=red,style=Dash](E)(G)
\Edge[color=black](E)(H)

\Vertex[x=30,y=20,Math,L={A},Lpos=90]{A2}
\Vertex[x=40,y=20,Math,L={B},Lpos=90]{B2}
\Vertex[x=45,y=10,Math,L={C},Lpos=60]{C2}
\Vertex[x=35,y=10,Math,L={D},Lpos=180]{D2}
\Vertex[x=50,y=10,Math,L={E},Lpos=0]{E2}
\Vertex[x=35,y=0,Math,L={F},Lpos=180]{F2}
\Vertex[x=55,y=0,Math,L={G},Lpos=180]{G2}
\Vertex[x=55,y=20,Math,L={H},Lpos=90]{H2}
\AddVertexColor{black}{B2,D2,E2,F2}
\Edge[style=Dash](A2)(B2)
\Edge[color=red,style=Dash](B2)(C2)
\Edge[style=Dash](C2)(D2)
\Edge[color=green,style=Dash](C2)(E2)
\Edge[color=blue](C2)(F2)
\Edge[color=red,style=Dash](E2)(G2)
\Edge[color=black](E2)(H2)

\end{tikzpicture}
\end{center}
\caption{Before and after the vertex switch at vertex $C$.  Note that only a portion of the Cliffordinkra has been drawn here.}
\label{fig:vswitch}
\end{figure}

\begin{proposition}
A vertex switch on a Cliffordinkra results in another Cliffordinkra.
\end{proposition}
\begin{proof}
Suppose we do a vertex switch at a vertex $v$.  The only thing that needs to be checked is that the Totally odd condition holds.  Consider a bicolor cycle.  It either involves $v$ or it does not.  If it does not involve $v$ then the dashing on the edges of the bicolor cycle are not affected.  If it does involve $v$ then two edges in the bicolor cycle will change their dashing.  The number of dashed edges in the bicolor cycle by increase by 2, decrease by 2, or remain the same, but in any case it will remain odd.
\end{proof}

A vertex switch can be done on many vertices, one after another.  These operations commute, and the result is as follows: let $S$ be a set of vertices.  Given an edge $e$, we change the dashing on $e$ if and only if $e$ connects a vertex in $S$ with a vertex in the complement of $S$.

The reader can check that this is indeed the effect of doing a vertex switch on each of the elements of $S$.

\subsection{Equivalence of Cliffordinkras}
If $(V_1,E_1)$ is the set of vertices and the set of edges of one graph, and $(V_2,E_2)$ is the set of vertices and the set of edges for another graph, a {\em graph isomorphism} from the first graph to the second graph is a bijection $\phi_V:V_1\to V_2$ and a bijection $\phi_E:E_1\to E_2$ so that if $e$ is an edge in $E_1$ connecting vertices $v$ and $w$ in $V_1$, then $\phi_E(e)$ is an edge in $E_2$ connecting vertices $\phi_V(v)$ and $\phi_V(w)$ in $V_2$.  This much is standard terminology.

If the vertices are bipartitioned into bosons and fermions, and the edges are colored with various colors, we can furthermore insist that these bijections respect these features.  To wit, if $v$ is a boson, then $\phi_V(v)$ is a boson, and if $v$ is a fermion, then $\phi_V(v)$ is a fermion.  Likewise the edge $e\in E_1$ and $\phi_E(e)\in E_2$ have the same color.

When it comes to dashing, however, we will take a different approach to include vertex switching.  So a Cliffordinkra isomorphism is a graph isomorphism that sends bosons to bosons and fermions to fermions, that sends edges to edges of the same color, and so that there exists a set $S\subset V_1$ of vertices so that for every edge $e\in E_1$, $\phi_E(e)$ has the same dashing as $e$ if and only if $e$ connects a vertex in $S$ to a vertex not in $S$.

\section{Clifford algebra representations from Cliffordinkras}
\label{sec:cliffordinkratorep}
In this section we will show how a Cliffordinkra with $n$ colors gives rise to a real $\zt$-graded representation of a Clifford algebra, more specifically, the Clifford algebra $\Cl(0,n)$.

Suppose we have a Cliffordinkra with $n$ colors.  Let $V$ be the free vector space spanned by the vertices of the graph; i.e., $V$ is the set of formal expressions of the form
\[a_1v_1+\cdots+a_kv_k\]
where $a_1, \ldots, a_k\in\re$ and $v_1,\ldots,v_k$ are vertices of the graph.

We now define linear functions $\Gamma_1,\ldots,\Gamma_n:V\to V$ as follows.  
If $i$ is a color, and $v$ and $w$ are connected by an edge of color $i$, then $\Gamma_i(v)=w$ if the edge is solid, and $\Gamma_i(v)=-w$ if it is dashed.

Note by the Regular edge-colored property that $\Gamma_i$ is well-defined.  These extends uniquely to linear functions $\Gamma_1,\ldots,\Gamma_n$.  If we write $V^0$ for the vector space spanned by the bosons and $V^1$ for the vector space spanned by the fermions, so that $V=V^0\oplus V^1$, we see that $\Gamma_i$ maps $V^0$ to $V^1$ and vice-versa.

By definition,
\[\Gamma_i\Gamma_i=I\]
where $I$ is the identity.  The Quadrilateral and Totally odd dashing properties show that if $i\not=j$,
\[\Gamma_i\Gamma_j=-\Gamma_j\Gamma_i\]
These properties, put together, can be written as an anticommutation relation
\begin{equation}
\{\Gamma_i,\Gamma_j\}=2\delta_{ij}I
\label{eqn:clifford}
\end{equation}
where $\delta_{ij}$ is the Kronecker delta symbol.

The $\Gamma_1,\ldots,\Gamma_n$ thus furnish a real $\zt$-graded representation of the Clifford algebra $\Cl(0,n)$.

We can actually get $\Gamma_i$ to be matrices if we number the vertices of the Cliffordinkra.  Then we can view $V$ as $\mathbb{R}^{\#\{\mbox{vertices}\}}$ where the $j$th vertex is associated with the standard basis element $e_j$.  Then the $(j,k)$ entry for $\Gamma_i$ is:
\begin{equation}
(\Gamma_i)_{j,k}=\begin{cases}
1,&\mbox{if there is a solid edge of color $i$ from vertex $j$ to vertex $k$}\\
-1,&\mbox{if there is a dashed edge of color $i$ from vertex $j$ to vertex $k$}\\
0,&\mbox{if there is no edge of color $i$ from vertex $j$ to vertex $k$.}
\end{cases}
\label{eqn:gammapm1}
\end{equation}
These give us matrices $\Gamma_1,\ldots,\Gamma_n$, again satisfying he Clifford algebra relations for $\Cl(0,n)$.  The entries are real numbers and in fact are $0$, $1$, and $-1$, and satisfy a further property: they are signed permutation matrices.  We define and explore signed permutation matrices in the next section.

\section{$\zt$-graded signed permutation representations}
\label{sec:signedperm}
A permutation on $m$ elements is a bijection on a set of $m$ elements.  If we use as this set the standard basis $e_1,\ldots,e_m$ in $\re^m$, and extend this map linearly, we can consider the permutation as a linear transformation on $\re^m$, and thus, using this standard basis, get a matrix.  This matrix will have in every row and in every column one nonzero entry, which is $1$.  Such a matrix is called a permutation matrix.

A {\em signed permutation matrix} is a square matrix with the property that every row and every column has one nonzero entry, which is either $1$ or $-1$. For instance, the following is a signed permutation matrix:
\[\begin{bmatrix}
-1&0&0\\
0&0&1\\
0&-1&0
\end{bmatrix}\]
As a linear transformation it sends each $e_i$ to either some $e_j$, or some $-e_j$.

Given a signed permutation matrix $P$, we can construct a permutation matrix $|P|$ by taking the absolute value of each entry.  Then if $P$ sent $e_i$ to $e_j$ or $-e_j$, $|P|$ would send $e_i$ to $e_j$.  Since matrix multiplication encodes function composition, it follows that $|PQ|=|P||Q|$ for any signed permutation matrices $P$ and $Q$.  Every signed permutation matrix $P$ can be written as $P=S\pi$, where $\pi$ is a permutation matrix (and in fact $\pi=|P|$) and $S$ is a diagonal matrix where the diagonal entries are all $1$ and $-1$.

The relevance of signed permutation matrices to Cliffordinkras is that this is how $\Gamma_i$ acts on the vertices of the Cliffordinkra, as elements of the vector space $V$ spanned by the vertices.  More specifically, a numbering of the vertices from 1 to the number of vertices gives rise to an isomorphism of $V$ to $\re^n$, so that each vertex goes to some $e_i$.  Under this isomorphism, $\Gamma_i$ is then a signed permutation matrix.

A {\em signed permutation representation} $\phi$ of the Clifford algebra $\Cl(0,n)$ of dimension $k$ means a Clifford algebra representation where all $\phi(\Gamma_i)$ are $k\times k$ signed permutation matrices.

A $\zt$-graded reprsentation of $\Cl(0,n)$ means a representation where the $\phi(\Gamma_i)$ are of the form
\[\phi(\Gamma_i)=\begin{bmatrix}
0&L_i\\
R_i&0
\end{bmatrix}\]
where $L_i$ and $R_i$ are $k/2\times k/2$ square matrices.  Another way to think of this is that $\re^k$ splits as a direct sum $\re^{k/2}\oplus \re^{k/2}$ and each $\Gamma_i$ is required to send vectors in the first summand to vectors in the second, and vice-versa.

A $\zt$-graded signed permutation representation is one that is both $\zt$-graded and where the $\phi(\Gamma_i)$ are signed permutation matrices.

Two $\zt$-graded signed permutation representations $\phi_1$ and $\phi_2$ are {\em $\zt$-graded signed permutation isomorphic} if there is a signed permutation $P$ of the form
\[P=\begin{bmatrix}
X&0\\
0&Y\end{bmatrix}\]
where $X$ and $Y$ are two $k/2\times k/2$ square matrices, so that
\begin{equation}
\phi_2(\Gamma_i)=P\phi_1(\Gamma_i)P^{-1}
\label{eqn:isorep}
\end{equation}
for all $\Gamma_i$.

Using Section~\ref{sec:cliffordinkratorep}, every Cliffordinkra with $n$ colors gives rise to a $\zt$-graded signed permutation representation of $\Cl(0,n)$.  In other words, if the Cliffordinkra has $k$ vertices, then the construction gives a set of square $k\times k$ matrices $\Gamma_1,\ldots,\Gamma_n$.  By (\ref{eqn:gammapm1}), these are signed permutation matrices.  The set of bosons and the set of fermions each span subspaces $V^0$ and $V^1$, and each $\Gamma_i$ sends $V^0$ to $V^1$ and vice-versa, because every edge of the Cliffordinkra connects a boson to a fermion.

Conversely, given a $\zt$-graded signed permutation representation of $\Cl(0,n)$ as $k\times k$ matrices, we can draw $k/2$ bosons (i.e., open dots) and $k/2$ fermions (i.e, filled dots), numbered $v_1$ through $v_k$.  We then draw an edge of color $c$ between vertex $v_i$ and $v_j$ if $(\Gamma_c)_{ij}$ is nonzero: a solid edge if it is $1$; a dashed edge if it is $-1$.  It is straightforward to prove that the definition of Cliffordinkra is satisfied.

We will compare the definition of isomorphism for Cliffordinkras and for $\zt$-graded signed permutation representations.  As a prelude to this, suppose we have a Cliffordinkra and its corresponding permutation representation, denoted by $\phi_1$.  Consider the effect of a signed permutation isomorphism given by the following signed permutation matrix:
\[P=\begin{bmatrix}
-1&0&\cdots&0\\
0&1&\cdots&0\\
\vdots&&&\vdots\\
0&0&\cdots&1
\end{bmatrix}.\]
The effect of (\ref{eqn:isorep}) is to modify the $\Gamma_i$ matrices by reversing the signs on the first column, and then reversing the signs on the first row.  This has the effect of modifying the Cliffordinkra by changing the dashing on all edges involving vertex $1$.  In other words, a vertex switch on vertex $1$.  More generally, if $P$ is a diagonal signed permutation matrix, then the isomorphism given by (\ref{eqn:isorep}) has the effect of a vertex switch on the vertices corresponding to the locations of the $-1$s on the diagonal of $P$.

We are now ready to explain the relationship between Cliffordinkra isomorphisms and $\zt$-graded signed permutation isomorphisms.  Just as a Cliffordinkra isomorphism can be understood as a graph isomorphism, followed by a vertex switch on a set of vertices, in the same way, a $\zt$-graded signed permutation isomorphism is a permutation matrix followed by a diagonal signed permutation matrix.

In this way, Cliffordinkras up to Cliffordinkra isomorphisms are the same as $\zt$-graded signed permutation representations of $\Cl(0,n)$, up to $\zt$-graded signed permutation isomorphisms.

\section{Cubes and Quotients of cubes}
\label{sec:cubes}
It is now time to introduce the basic constructions from which all Cliffordinkras with $n$ colors can be built.  Back in Section~\ref{sec:cliffordinkraex} we constructed Cliffordinkras that looked like a line segment for $n=1$, a square for $n=2$, and a three-dimensional cube (at least, its $1$-skeleton) for $n=3$.

More generally, we consider an $n$-dimensional cube, which we can take to be $[0,1]^n\subset \mathbb{R}^n$.  The set of vertices of this $n$-cube is $\{0,1\}^n$; in other words, a vertex is an $n$-tuple of $0$s and $1$s.  We say that two vertices are connected by an edge of color $i$ if the vertices agree in every coordinate except in coordinate $i$, that is, the vertices are of the form
\begin{gather*}
(x_1,\ldots,x_{i-1},0,x_{i+1},\ldots,x_n)\\
(x_1,\ldots,x_{i-1},1,x_{i+1},\ldots,x_n).
\end{gather*}
The vertex $(x_1,\ldots,x_n)$ is a boson if $\sum x_i$ is even, and a fermion if $\sum x_i$ is odd.

The reader can check that this satisfies the Bipartite, Regular edge-colored, and Quadrilateral properties.  A more difficult challenge is to prove that there exists a dashing that satisfies the Totally odd property, and this fact will be a consequence of the ideas in Section~\ref{sec:cliffordcube}.

This construction gives rise to Cliffordinkras for any $n$, and these Cliffordinkras we call {\em cubical Cliffordinkras}.  But there are connected Cliffordinkras that are not cubical: the Cliffordinkra from Figure~\ref{fig:cliffordinkraex} has $n=4$ colors, and $8$ vertices, while a 4-dimensional cube has $2^4=16$ vertices.

That Cliffordinkra is obtained as follows.  We begin with the $4$-dimensional cube:

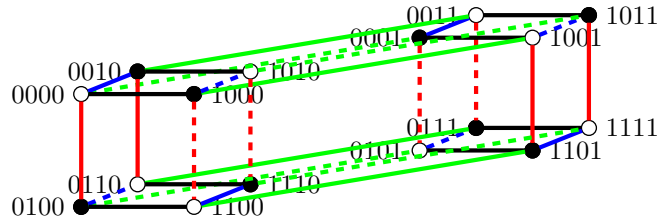
\begin{center}
\begin{tikzpicture}[scale=0.15]
\GraphInit[vstyle=Welsh]
\SetVertexNormal[MinSize=5pt]
\SetUpEdge[labelstyle={draw},style={ultra thick}]
\tikzset{Dash/.style={dashed,draw,ultra thick}}
\Vertex[x=30,y=5,Math,L={0001},Lpos=180]{0001}
\Vertex[x=40,y=5,Math,L={1001},Lpos=0]{1001}
\Vertex[x=30,y=-5,Math,L={0101},Lpos=180]{0101}
\Vertex[x=40,y=-5,Math,L={1101},Lpos=0]{1101}
\Vertex[x=35,y=7,Math,L={0011},Lpos=180]{0011}
\Vertex[x=45,y=7,Math,L={1011},Lpos=0]{1011}
\Vertex[x=35,y=-3,Math,L={0111},Lpos=180]{0111}
\Vertex[x=45,y=-3,Math,L={1111},Lpos=0]{1111}

\Vertex[x=0,y=0,Math,L={0000},Lpos=180]{0000}
\Vertex[x=10,y=0,Math,L={1000},Lpos=0]{1000}
\Vertex[x=0,y=-10,Math,L={0100},Lpos=180]{0100}
\Vertex[x=10,y=-10,Math,L={1100},Lpos=0]{1100}
\Vertex[x=5,y=2,Math,L={0010},Lpos=180]{0010}
\Vertex[x=15,y=2,Math,L={1010},Lpos=0]{1010}
\Vertex[x=5,y=-8,Math,L={0110},Lpos=180]{0110}
\Vertex[x=15,y=-8,Math,L={1110},Lpos=0]{1110}
\AddVertexColor{black}{1000,0100,0010,0001,1110,1101,1011,0111}
\Edge(0011)(1011)
\Edge(0111)(1111)
\Edge[color=red,style=Dash](0011)(0111)
\Edge[color=red](1011)(1111)
\Edge[color=blue](0001)(0011)
\Edge[color=blue,style=Dash](1001)(1011)
\Edge[color=blue,style=Dash](0101)(0111)
\Edge[color=blue](1101)(1111)
\Edge(0001)(1001)
\Edge(0101)(1101)
\Edge[color=red,style=Dash](0001)(0101)
\Edge[color=red](1001)(1101)

\Edge[color=green,style=Dash](0000)(0001)
\Edge[color=green](1000)(1001)
\Edge[color=green,style=Dash](0100)(0101)
\Edge[color=green](1100)(1101)
\Edge[color=green](0010)(0011)
\Edge[color=green,style=Dash](1010)(1011)
\Edge[color=green](0110)(0111)
\Edge[color=green,style=Dash](1110)(1111)
\Edge(0010)(1010)
\Edge(0110)(1110)
\Edge[color=red](0010)(0110)
\Edge[color=red,style=Dash](1010)(1110)
\Edge[color=blue](0000)(0010)
\Edge[color=blue,style=Dash](1000)(1010)
\Edge[color=blue,style=Dash](0100)(0110)
\Edge[color=blue](1100)(1110)
\Edge(0000)(1000)
\Edge(0100)(1100)
\Edge[color=red](0000)(0100)
\Edge[color=red,style=Dash](1000)(1100)
\end{tikzpicture}
\end{center}

We then pair together vertices according to the rule that $(x_1,x_2,x_3,x_4)$ gets paired with $(1-x_1,1-x_2,1-x_3,1-x_4)$.  So for instance, $0010$ gets paired with $1101$.  The pairs are shown below, linked with purple edges:

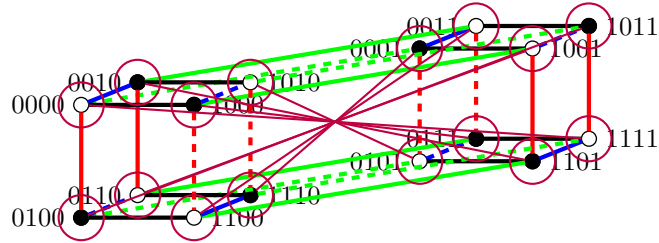
\begin{center}
\begin{tikzpicture}[scale=0.15]
\GraphInit[vstyle=Welsh]
\SetVertexNormal[MinSize=5pt]
\SetUpEdge[labelstyle={draw},style={ultra thick}]
\tikzset{Dash/.style={dashed,draw,ultra thick}}
\Vertex[x=30,y=5,Math,L={0001},Lpos=180]{0001}
\Vertex[x=40,y=5,Math,L={1001},Lpos=0]{1001}
\Vertex[x=30,y=-5,Math,L={0101},Lpos=180]{0101}
\Vertex[x=40,y=-5,Math,L={1101},Lpos=0]{1101}
\Vertex[x=35,y=7,Math,L={0011},Lpos=180]{0011}
\Vertex[x=45,y=7,Math,L={1011},Lpos=0]{1011}
\Vertex[x=35,y=-3,Math,L={0111},Lpos=180]{0111}
\Vertex[x=45,y=-3,Math,L={1111},Lpos=0]{1111}
\Vertex[x=0,y=0,Math,L={0000},Lpos=180]{0000}
\Vertex[x=10,y=0,Math,L={1000},Lpos=0]{1000}
\Vertex[x=0,y=-10,Math,L={0100},Lpos=180]{0100}
\Vertex[x=10,y=-10,Math,L={1100},Lpos=0]{1100}
\Vertex[x=5,y=2,Math,L={0010},Lpos=180]{0010}
\Vertex[x=15,y=2,Math,L={1010},Lpos=0]{1010}
\Vertex[x=5,y=-8,Math,L={0110},Lpos=180]{0110}
\Vertex[x=15,y=-8,Math,L={1110},Lpos=0]{1110}
\AddVertexColor{black}{1000,0100,0010,0001,1110,1101,1011,0111}
\Edge(0011)(1011)
\Edge(0111)(1111)
\Edge[color=red,style=Dash](0011)(0111)
\Edge[color=red](1011)(1111)
\Edge[color=blue](0001)(0011)
\Edge[color=blue,style=Dash](1001)(1011)
\Edge[color=blue,style=Dash](0101)(0111)
\Edge[color=blue](1101)(1111)
\Edge(0001)(1001)
\Edge(0101)(1101)
\Edge[color=red,style=Dash](0001)(0101)
\Edge[color=red](1001)(1101)

\Edge[color=green,style=Dash](0000)(0001)
\Edge[color=green](1000)(1001)
\Edge[color=green,style=Dash](0100)(0101)
\Edge[color=green](1100)(1101)
\Edge[color=green](0010)(0011)
\Edge[color=green,style=Dash](1010)(1011)
\Edge[color=green](0110)(0111)
\Edge[color=green,style=Dash](1110)(1111)
\Edge(0010)(1010)
\Edge(0110)(1110)
\Edge[color=red](0010)(0110)
\Edge[color=red,style=Dash](1010)(1110)
\Edge[color=blue](0000)(0010)
\Edge[color=blue,style=Dash](1000)(1010)
\Edge[color=blue,style=Dash](0100)(0110)
\Edge[color=blue](1100)(1110)
\Edge(0000)(1000)
\Edge(0100)(1100)
\Edge[color=red](0000)(0100)
\Edge[color=red,style=Dash](1000)(1100);
\draw[color=purple,thick] (0000) circle (2);
\draw[color=purple,thick] (0100) circle (2);
\draw[color=purple,thick] (1000) circle (2);
\draw[color=purple,thick] (1100) circle (2);
\draw[color=purple,thick] (0010) circle (2);
\draw[color=purple,thick] (0110) circle (2);
\draw[color=purple,thick] (1010) circle (2);
\draw[color=purple,thick] (1110) circle (2);
\draw[color=purple,thick] (0001) circle (2);
\draw[color=purple,thick] (0101) circle (2);
\draw[color=purple,thick] (1001) circle (2);
\draw[color=purple,thick] (1101) circle (2);
\draw[color=purple,thick] (0011) circle (2);
\draw[color=purple,thick] (0111) circle (2);
\draw[color=purple,thick] (1011) circle (2);
\draw[color=purple,thick] (1111) circle (2);
\draw[color=purple,thick] (0000)--(1111);
\draw[color=purple,thick] (0100)--(1011);
\draw[color=purple,thick] (1000)--(0111);
\draw[color=purple,thick] (1100)--(0011);
\draw[color=purple,thick] (0010)--(1101);
\draw[color=purple,thick] (0110)--(1001);
\draw[color=purple,thick] (1010)--(0101);
\draw[color=purple,thick] (1110)--(0001);
\end{tikzpicture}
\end{center}

We then form a Cliffordinkra by taking these pairs of vertices for our new vertices.  Since every pair has one vertex whose last bit is 0 (in the cube in the lower left) and one vertex whose last bit is 1 (in the cube in the upper right), we can take as a representative the vertex whose last bit is 0.  This gives us the cube in the lower left.

There is an edge in the new Cliffordinkra of color $i$ between $v$ and $w$ if there is an edge of that color between either of the vertices that produced $v$ to either of the vertices that produced $w$.  The result is the following Cliffordinkra.  Again, we will postpone discussion of the dashed edges.

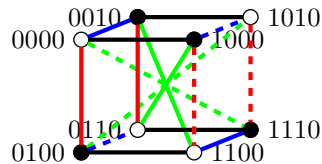
\begin{center}
\begin{tikzpicture}[scale=0.15]
\GraphInit[vstyle=Welsh]
\SetVertexNormal[MinSize=5pt]
\SetUpEdge[labelstyle={draw},style={ultra thick}]
\tikzset{Dash/.style={dashed,draw,ultra thick}}
\Vertex[x=0,y=0,Math,L={0000},Lpos=180]{0000}
\Vertex[x=10,y=0,Math,L={1000},Lpos=0]{1000}
\Vertex[x=0,y=-10,Math,L={0100},Lpos=180]{0100}
\Vertex[x=10,y=-10,Math,L={1100},Lpos=0]{1100}
\Vertex[x=5,y=2,Math,L={0010},Lpos=180]{0010}
\Vertex[x=15,y=2,Math,L={1010},Lpos=0]{1010}
\Vertex[x=5,y=-8,Math,L={0110},Lpos=180]{0110}
\Vertex[x=15,y=-8,Math,L={1110},Lpos=0]{1110}
\AddVertexColor{black}{1000,0100,0010,1110}
\Edge[color=green,style=Dash](0000)(1110)
\Edge[color=green](1000)(0110)
\Edge[color=green,style=Dash](0100)(1010)
\Edge[color=green](1100)(0010)

\Edge(0010)(1010)
\Edge(0110)(1110)
\Edge[color=red](0010)(0110)
\Edge[color=red,style=Dash](1010)(1110)
\Edge[color=blue](0000)(0010)
\Edge[color=blue,style=Dash](1000)(1010)
\Edge[color=blue,style=Dash](0100)(0110)
\Edge[color=blue](1100)(1110)
\Edge(0000)(1000)
\Edge(0100)(1100)
\Edge[color=red](0000)(0100)
\Edge[color=red,style=Dash](1000)(1100)
\end{tikzpicture}
\end{center}

More generally, suppose $C$ is a subgroup of $\zt^n=\{0,1\}^n$.  Take as the set of vertices the quotient $\zt^n/C$.  Then a vertex is a coset $v+C$, that is, a set of vertices in the $n$-dimensional cube so that any two differ by an element of $C$.  An edge of color $i$ in $\zt^n/C$ connects vertices $v+C$ and $w+C$ if there is an edge of color $i$ in $\zt^n$ connecting some vertex in $v+C$ with some vertex in $w+C$.  The example we just did above uses the subgroup $C=\{0000,1111\}\subset \zt^4$.

It turns out that every Cliffordinkra has the topology of a quotient of an $n$-dimensional cube by a subgroup $C$ of $\zt^n$, as we will see in Section~\ref{sec:at}. 

\section{Binary Linear Block Codes}
\label{sec:code}
From the previous section, we can see that it is important to study the subgroups of $\zt^n$.  This is the theory of binary linear block codes.  We include the basics here to make this paper self-contained, but a more thorough introduction is Ref.~\cite{rCHVP}.

In coding theory, two people communicate by sending strings of symbols.  A set of strings is called a {\em code}.  The idea is that these strings are those that are meaningful to the people communicating.  An element of the code is called a {\em codeword}.  One application is that if the communication method can sometimes garble the signal, then whenever a string that is not a codeword is received, then the recipient knows that an error has occurred and can ask for the message again.  This is called {\em error detection}.  Under certain circumstances it may be possible for the recipient to figure out what the message must have been, and to correct it without having to ask the recipient.  This is called {\em error correction}.  Error correction and error detection have many applications, including banking, wifi, Universal Product Codes, streaming video, computer storage, and much more.

Computers naturally use binary strings, that is, strings of $0$s and $1$s.  A code using only $0$s and $1$s is called a {\em binary code}.  A {\em block code} is a code where each string has the same length.  So for instance, a code like $\{0100,1111,0110\}$ is a binary block code of length 4.  A {\em linear code} is a binary block code that is closed under addition modulo 2.  That is, a binary linear block code of length $n$ is a subgroup of $\zt^n$.  So $\{000,011,101,110\}$ is a linear code because the sum of any two elements is another element.  We can also think of $\zt^N$ as a vector space over the field $\gf_2=\{0,1\}$ of integers modulo 2.  Then a linear code is a vector subspace of $\zt^N$.  It is therefore generated by a basis of nonzero codewords.  The number of elements in that basis is the {\em dimension} of the code.  For instance, the linear code $\{000,011,101,110\}$ is generated by the basis $\{011,101\}$.  A linear code of dimension $k$ has $2^k$ codewords.  It is sometimes convenient to write a basis as a matrix, where each row of the matrix is an element of the basis.  So the previous example has generating matrix
\[\begin{bmatrix}
0&1&1\\
1&0&1
\end{bmatrix}\]

The {\em weight} of a string $w$, written $\wt(w)$, is the number of $1$s.  So the weight of 1110101 is 5.  A code is {\em even} if every codeword has even weight, and a code is {\em doubly even} if every codeword has weight a multiple of $4$.  The {\em minimal weight} of a linear code is the minimum weight among nonzero codewords.

A fundamental property in linear codes is
\begin{equation}
\wt(v+w)=\wt(v)+\wt(w)-2\wt(v\wedge w)
\label{eqn:wtadd}
\end{equation}
where $v\wedge w$ is the string obtained by bitwise and of $v$ and $w$; that is, the $i$th bit of $v\wedge w$ is $1$ if and only if the $i$th bit of $v$ is $1$ and the $i$th bit of $w$ is $1$.

As a consequence, a code generated by $\{g_1,\ldots,g_k\}$ is even if and only if every $\wt(g_i)$ is even, and it is doubly even if and only if every $\wt(g_i)$ is a multiple of $4$, and for every $i$ and $j$, $g_i\wedge g_j$ has even weight.

\section{Classification of Cliffordinkra Topologies}
\label{sec:at}
We are now ready for the main classification of Cliffordinkra topologies.
\begin{thm}
Every connected Cliffordinkra has the topology of a quotient $\zt^n/C$, where $C$ is a doubly even binary linear block code.
\end{thm}
\begin{proof}
Suppose one has a connected Cliffordinkra.  This has a set of vertices and matrices $\Gamma_1,\ldots,\Gamma_n$.

Let $|\Gamma_i|$ be the result of taking the absolute value of each entry of $\Gamma_i$.  Then $|\Gamma_i|$ can be viewed as a permutation of the vertices of the Cliffordinkra.  Recall from Section~\ref{sec:signedperm} that $|AB|=|A||B|$ when $A$ and $B$ are signed permutation matrices.  Then $|\Gamma_i|^2=I$ and for any $i\not=j$, $|\Gamma_i||\Gamma_j|=|\Gamma_j||\Gamma_i|$. 

 From this, $\zt^N$ acts on the set of vertices as follows:
\[(x_1,\ldots,x_n)\cdot v = |\Gamma_1|^{x_1}\cdots|\Gamma_n|^{x_n}(v)\]
To see that this is a group action, we compute:
\begin{align*}
(x_1,\ldots,x_n)\cdot((y_1,\ldots,y_n)\cdot v)
&= |\Gamma_1|^{x_1}\cdots|\Gamma_n|^{x_n} |\Gamma_1|^{y_1}\cdots|\Gamma_n|^{y_n}(v)\\
|\Gamma_1|^{x_1+y_1}\cdots|\Gamma_n|^{x_n+y_n}(v)\\
\end{align*}
where we have used the fact that the $|\Gamma_i|$ commute.  Since $|\Gamma_i|^2=I$, we can compute $x_i+y_i$ modulo 2, and we get
\[(x_1+y_1,\ldots,x_n+y_n)\cdot v.\]
thereby demonstrating that this is a group action of $\zt^n$ on the set of vertices.

We now show that if the Cliffordinkra is connected, then this action is transitive.  Let $v$ and $w$ be vertices.  Since the Cliffordinkra is connected, there is a path from $v$ to $w$.  Let $k$ be the length of this path.  We will prove that there is a sequence of $\Gamma_{i_1},\ldots,\Gamma_{i_k}$ so that $|\Gamma_{i_k}|\cdots|\Gamma_{i_1}|(v)=w$.  This is proved using induction on $k$.

The base case $k=0$ is trivial: $v=w$ and we use the empty sequence.  The inductive case involves the length of a path from $v$ to $w$ being $k>0$.  Then this path is a concatenation of a path from $v$ to some vertex $x$ of length $k-1$, followed by an edge $e$ from $x$ to $w$.  Let the color of $e$ be $i$, and let $e_i$ be the $n$-tuple $e_i=(0,\ldots,0,1,0,\ldots,0)$ with a $1$ in the $i$th entry.  The inductive hypothesis says that there is a $(x_1,\ldots,x_n)$ so that $(x_1,\ldots,x_n)\cdot v=x$.  Since $e_i\cdot x=|\Gamma_i|x=w$, we get $(e_i+(x_1,\ldots,x_n))\cdot v=w$.

  Pick a vertex $v_0$.  Let $C$ be the stabilizer of $v_0$ (in fact, it can be proven straightforwardly that $C$ does not depend on $v_0$).  This $C$ is a binary linear block code, and is the code associated to a Cliffordinkra.  By the Orbit Stabilizer theorem, we have a bijection from $\zt^n/C$ to the set of vertices.  This sends $(x_1,\ldots,x_n)$ to $(x_1,\ldots,x_n)\cdot v_0$.  If there is an edge in the Cliffordinkra of color $i$ connecting $v$ to $w$, then $|\Gamma_i|(v)=w$, and so $e_i\cdot v=w$, where $e_i$ is the standard basis element in $\zt^n$.  If $v=(x_1,\ldots,x_n)\cdot v_0$, then $w=(e_i+(x_1,\ldots,x_n))\cdot v_0$.  Then, there is a edge in $\zt^n$ of color $i$ between the the vertices in $\zt^n/C$ corresponding to $v$ and $w$ by this bijection.  Therefore there is a bijection between the Cliffordinkra and $\zt^n/C$ for some binary linear block code $C$ of length $n$, and this preserves edges and colors.

Since each $\Gamma_i$ connects bosons to fermions and vice-versa, products of $\Gamma_i$ connect bosons to bosons if and only if there are an even number of $\Gamma_i$s in that product.  Then if
\[(x_1,\ldots,x_n)\cdot v=v\]
then $\wt(x_1,\ldots,x_n)$ must be even.  Thus, $C$ must be even.

To show $C$ must be doubly even, let $x=(x_1,\ldots,x_n)\in C$.  Then
\[|\Gamma_1|^{x_1}\cdots|\Gamma_n|^{x_n}v_0=v_0.\]
Now $\Gamma_1{}^{x_1}\cdots\Gamma_n{}^{x_n}$ is a signed permutation matrix, and its absolute value, as we just saw, sends $v_0$ to $v_0$.  Therefore,
\[\Gamma_1{}^{x_1}\cdots\Gamma_n{}^{x_n}v_0=\pm v_0.\]
Whichever it is, when we do it twice, we come back to where we started:
\begin{equation}
\left(\Gamma_1{}^{x_1}\cdots\Gamma_n{}^{x_n}\right)^2v_0= v_0.
\label{eqn:loopdash}
\end{equation}
We can then simplify
\[\left(\Gamma_1{}^{x_1}\cdots\Gamma_n{}^{x_n}\right)^2
=\left(\Gamma_1{}^{x_1}\cdots\Gamma_n{}^{x_n}\right)
\left(\Gamma_1{}^{x_1}\cdots\Gamma_n{}^{x_n}\right)\]
by moving the $\Gamma_i$ past each other, gaining a $-1$ every time we do so.  The first nonzero $x_i$ moves past $\wt(x)-1$ other $\Gamma_i$.  The next moves past $\wt(x)-2$ other $\Gamma_i$.  This continues until we have moved a total of
\[\frac{\wt(x)(\wt(x)-1)}{2}\]
times.  To make (\ref{eqn:loopdash}) correct, this number must be even.  Since $\wt(x)$ is already even, $\wt(x)-1$ is odd, and this fraction can only be even if $\wt(x)$ is a multiple of $4$.  Therefore $C$ is doubly even.
\end{proof}

Conversely, every quotient $\zt^n/C$, where $C$ is a doubly even binary linear block code, is the topology of a connected Cliffordinkra.  The proof of this is in the next section.

In the meantime, here is a chart illustrating the relationship between weights of codewords and properties of the quotient a cube by a code:
\begin{center}
\begin{tabular}{ccc}
Weight of&&Quotient\\
codeword&&$\zt^n/C$\\\hline
even&$\leftrightarrow$&bipartite\\
$\ne 1, 2$&$\leftrightarrow$&quadrilateral\\
doubly even&$\leftrightarrow$&totally odd
\end{tabular}
\end{center}
If all weights of $C$ are even, then the quotient $\zt^n/C$ is bipartite, and vice-versa.  If all weights of $C$ are not $1$ or $2$, then $\zt^n/C$ has the quadrilateral property, and conversely.  This is not hard to see: the cube has the quadrilateral property and codewords of weight $1$ and $2$ are precisely the ones that will turn bicolor cycles into cycles of length $1$ and $2$.  What we have just seen is that if $C$ is doubly even, then $\zt^n/C$ has a totally odd dashing.  In the next section, we show how, given a doubly even code $C$, to construct an Adinkra, and the corresponding $\zt$-graded signed permutation representaion of $\Cl(0,n)$.

\section{The Clifford Cube and its quotients}
\label{sec:cliffordcube}
In this section we will use Clifford algebra constructions to create Cliffordinkras and their corresponding signed permutation representations of $\Cl(0,n)$.  This will construct all such Cliffordinkras up to equivalence, and thus, all signed permutation representations up to isomorphism.

First, we construct the $n$-cube.  To do this, we take the action of $\Cl(0,n)$ on itself by left multiplication.  We denote by $\Gamma_1,\ldots,\Gamma_n$ the generators of $\Cl(0,n)$.\footnote{We use the same symbol $\Gamma_i$ for the element of the Clifford algebra and for the matrix in its representation.  The context will make clear which is meant.}  For each $n$-tuple $(x_1,\ldots,x_n) \in \{0,1\}^n$, we define
\[\Gamma_{x_1,\ldots,x_n}=\Gamma_1{}^{x_1}\cdots\Gamma_n{}^{x_n}.\]
These products form a basis of $2^n$ elements of $\Cl(0,n)$, as a real vector space.

If we act with $\Gamma_i$ on such a product, the result will be another such product, but with the $i$th coordinate changed, perhaps with a minus sign, due to the definition of $\Gamma_{x_1,\ldots,x_n}$ and the anticommutativity of the $\Gamma_i$.  Thus, $\Gamma_i$ acts as a signed permutation of the $\Gamma_{x_1,\ldots,x_n}$.  If we write $e_i=(0,\ldots,0,1,0,\ldots,0)$ for the $i$th standard basis element in $\mathbb{R}^n$, then we might write this as
\[\Gamma_i\Gamma_{x_1,\ldots,x_n}=\pm \Gamma_{e_i+(x_1,\ldots,x_n)}.\]

We can associate to each $\Gamma_{x_1,\ldots,x_n}$ the $n$-tuple $(x_1,\ldots,x_n)\in \{0,1\}^n$, which is a vertex of the $n$-cube.  We draw a solid edge of color $i$ from $(x_1,\ldots,x_n)$ to $e_i+(x_1,\ldots,x_n)$ if $\Gamma_i\Gamma_{x_1,\ldots,x_n}=\Gamma_{e_i+(x_1,\ldots,x_n)}$, and a dashed edge of color $i$ if $\Gamma_i\Gamma_{x_1,\ldots,x_n}=-\Gamma_{e_i+(x_1,\ldots,x_n)}$.

In this way we end up with a Cliffordinkra with the topology of the $n$-cube.  This matches the construction in Section~\ref{sec:cliffordinkratorep}.  By the way, in particular, this shows that there are in fact totally odd dashings on the $n$-cube.

We now show how to perform a quotient of this $n$-cube.

If $x=(x_1,\ldots,x_n)\in \{0,1\}^n$, we define two Clifford algebra elements $\pi_{x,+}$ and $\pi_{x,-}$ as follows:
\[\pi_{x,\pm}=\frac12(1\pm \Gamma_x).\]
This $\pi_{x,\pm}$ is purely even if and only if $\wt(x)$ is even.  Now $\Gamma_x{}^2=\pm I$, where this is $I$ if $\wt(x)$ is $0$ or $1$ mod 4 and $-I$ if $\wt(x)$ is $2$ or $3$ mod $4$.  We will be interested in the case where $\wt(x)$ is a multiple of $4$.  In that situation,
\begin{align*}
\pi_{x,\pm}{}^2&=\frac14(1\pm \Gamma_x)^2\\
&=\frac14(1\pm 2\Gamma_x+1)\\
&=\frac12(1\pm \Gamma_x)\\
&=\pi_{x,\pm}
\end{align*}
so right multiplication by $\pi_{x,\pm}$ is a projection operator on $\Cl(0,n)$.  This is consistent with the left multiplication action of $\Cl(0,n)$ on itself because left multiplication commutes with right multiplication.

This splits $\Cl(0,n)$ as a direct sum of two subrepresentations: the image of $\pi_+$ which is the kernel of $\pi_-$; and the image of $\pi_-$ which is the kernel of $\pi_+$.

For example, take $n=4$, and $x=1111$.  There are sixteen vertices in the $4$-cube:
\[\Gamma_{0000},\Gamma_{0001},\Gamma_{0010},\Gamma_{0011},\]
\[\Gamma_{0100},\Gamma_{0101},\Gamma_{0110},\Gamma_{0111},\]
\[\Gamma_{1000},\Gamma_{1001},\Gamma_{1010},\Gamma_{1011},\]
\[\Gamma_{1100},\Gamma_{1101},\Gamma_{1110},\Gamma_{1111}\]
For each of these original sixteen vertices, we apply $\pi_{1111,+}$ and $\pi_{1111,-}$, so we should have 32 answers.  But we will see that there will be an overcounting: define $a_y=\Gamma_y\pi_{1111,+}$ and $b_y=\Gamma_y\pi_{1111,-}$, and we see all of these projections:
\begin{align*}
a_{0000}&=a_{1111}, &a_{0011}&=-a_{1100}, &a_{0101}&=a_{1010}, &a_{0110}&=-a_{1001}\\
a_{0001}&=-a_{1110}, &a_{0010}&=a_{1101}, &a_{0100}&=-a_{1011}, &a_{1000}&=a_{0111}\\
b_{0000}&=-b_{1111}, &b_{0011}&=b_{1100}, &b_{0101}&=-b_{1010}, &b_{0110}&=b_{1001}\\
b_{0001}&=b_{1110}, &b_{0010}&=-b_{1101}, &b_{0100}&=b_{1011}, &b_{1000}&=-b_{0111}\\
\end{align*}
As you can see, these pair up into 16 different answers.

We can apply $\Gamma_i$ as before to these $a_y$ and $b_y$ and because left multiplication commutes with right multiplication,
\[\Gamma_ia_y=\Gamma_i\Gamma_y\pi_{x,+}=\pm\Gamma_{e_i+y}\pi_{x,+}=\pm a_{e_i+y}\]
so $\Gamma_i$ sends an $a$-labeled vertex to another $a$-labeled vertex.  Likewise it sends a $b$-labeled vertex to a $b$-labeled vertex.

Therefore, we have a disjoint union of two Cliffordinkras: one with all of the $a$ vertices and one with all of the $b$ vertices.

These Cliffordinkras are shown in Figure~\ref{fig:twins}.
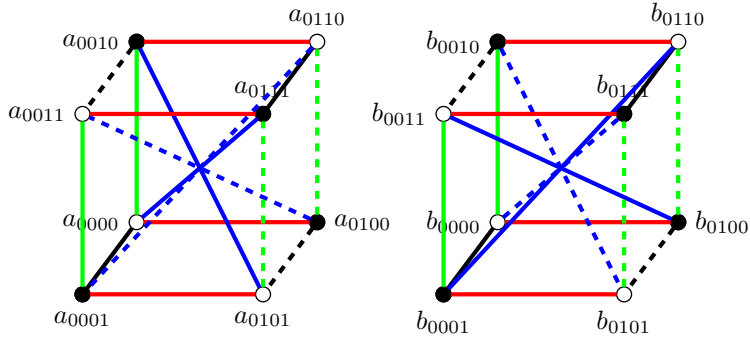
\begin{figure}[h]
\begin{center}
\begin{tikzpicture}[scale=0.12]
\GraphInit[vstyle=Welsh]
\SetVertexNormal[MinSize=5pt]
\SetUpEdge[labelstyle={draw},style={ultra thick}]
\tikzset{Dash/.style={dashed,draw,ultra thick}}
\Vertex[x=0,y=0,Math,L={a_{0000}},Lpos=180]{A}
\Vertex[x=20,y=0,Math,L={a_{0100}},Lpos=0]{C}
\Vertex[x=0,y=20,Math,L={a_{0010}},Lpos=180]{D}
\Vertex[x=20,y=20,Math,L={a_{0110}},Lpos=90]{E}
\Vertex[x=-6,y=-8,Math,L={a_{0001}},Lpos=270]{B}
\Vertex[x=14,y=-8,Math,L={a_{0101}},Lpos=270]{F}
\Vertex[x=-6,y=12,Math,L={a_{0011}},Lpos=180]{G}
\Vertex[x=14,y=12,Math,L={a_{0111}},Lpos=90]{H}
\AddVertexColor{black}{B,C,D,H}
\Edge[color=red](A)(C)
\Edge[color=red](D)(E)
\Edge[color=green](A)(D)
\Edge[color=green, style=Dash](C)(E)
\Edge(A)(B)
\Edge[style=Dash](C)(F)
\Edge[style=Dash](D)(G)
\Edge(E)(H)
\Edge[color=blue](A)(H)
\Edge[color=blue,style=Dash](C)(G)
\Edge[color=blue](D)(F)
\Edge[color=blue,style=Dash](E)(B)
\Edge[color=red](B)(F)
\Edge[color=red](G)(H)
\Edge[color=green](B)(G)
\Edge[color=green, style=Dash](F)(H)
\Vertex[x=40,y=0,Math,L={b_{0000}},Lpos=180]{A2}
\Vertex[x=60,y=0,Math,L={b_{0100}},Lpos=0]{C2}
\Vertex[x=40,y=20,Math,L={b_{0010}},Lpos=180]{D2}
\Vertex[x=60,y=20,Math,L={b_{0110}},Lpos=90]{E2}
\Vertex[x=34,y=-8,Math,L={b_{0001}},Lpos=270]{B2}
\Vertex[x=54,y=-8,Math,L={b_{0101}},Lpos=270]{F2}
\Vertex[x=34,y=12,Math,L={b_{0011}},Lpos=180]{G2}
\Vertex[x=54,y=12,Math,L={b_{0111}},Lpos=90]{H2}
\AddVertexColor{black}{B2,C2,D2,H2}
\Edge[color=red](A2)(C2)
\Edge[color=red](D2)(E2)
\Edge[color=green](A2)(D2)
\Edge[color=green, style=Dash](C2)(E2)
\Edge(A2)(B2)
\Edge[style=Dash](C2)(F2)
\Edge[style=Dash](D2)(G2)
\Edge(E2)(H2)
\Edge[color=blue,style=Dash](A2)(H2)
\Edge[color=blue](C2)(G2)
\Edge[color=blue,style=Dash](D2)(F2)
\Edge[color=blue](E2)(B2)
\Edge[color=red](B2)(F2)
\Edge[color=red](G2)(H2)
\Edge[color=green](B2)(G2)
\Edge[color=green, style=Dash](F2)(H2)
\end{tikzpicture}
\end{center}
\caption{The two Cliffordinkras obtained by the projections $\pi_{1111,\pm}$.}
\label{fig:twins}
\end{figure}

Each of these is an Adinkra with code $\{0000,1111\}$, and they differ only in their dashing.

More generally, if $x=(x_1,\ldots,x_n)$ has weight a multiple of 4, define $\pi_{x,\pi}=\frac12(1\pm\Gamma_x)$.  Take the $\Gamma_y$ basis for the vector space $\Cl(0,n)$ and define $a_y=\Gamma_y\pi_{x,+}$ and $b_y=\Gamma_y\pi_{x,-}$.  These will form bases for two subrepresentations of $\Cl(0,n)$, namely, $\Cl(0,n)\pi_{x,+}$ and $\Cl(0,n)\pi_{x,-}$.  Since $\Gamma_x\pi_{x,+}=\pi_{x,+}$, and $\Gamma_x\pi_{x,-}=-\pi_{x,-}$, there is an identification $a_x=a_0$ and $b_x=-b_0$.  Likewise by applying various $\Gamma_y$ to the left of these identifications, we get identifications of $a_y$ with $\pm a_{x+y}$, and $b_y$ with $\mp b_{x+y}$.

We might attempt to do this trick more than once: suppose $x=(x_1,\ldots,x_n)$ and $z=(z_1,\ldots,z_n)$ are $n$-tuples of $0$s and $1$s, each with weight a multiple of $4$.  Then again, $\pi_{x,\pm}$ and $\pi_{z,\pm}$ are projection operators when viewed as multiplying $\Cl(0,n)$ on the right. Define:
\begin{align*}
\pi_{x,z,+,+}&=\pi_{x,+}\pi_{z,+}\\
\pi_{x,z,+,-}&=\pi_{x,+}\pi_{z,-}\\
\pi_{x,z,-,+}&=\pi_{x,-}\pi_{z,+}\\
\pi_{x,z,-,-}&=\pi_{x,-}\pi_{z,-}
\end{align*}
If $\pi_{x,\pm}$ and $\pi_{z,\pi}$ commute, then these four combinations are all projections as well, and the resulting image is a subrepresentation of $\Cl(0,n)$.  On the other hand, these commute if and only if $\wt(x\wedge z)$ is even, which we already saw in (\ref{eqn:wtadd}) was a requirement for $x+z$ being doubly even.  So once we quotient with $x$ if $\{0,x\}$ is a doubly even code, and once we have done that, we can quotient with $z$ if $\{0,x,z,x+z\}$ is a doubly even code.

More generally, suppose $C$ is a doubly even code of length $n$, with generating set $\{g_1,\ldots,g_k\}$.  Let $(s_1,\ldots,s_k)\in\{+,-\}^k$ be a sequence of plus signs or minus signs.  Define
\[\pi_{(g_1,\ldots,g_k),(s_1,\ldots,s_k)}=\pi_{g_1,s_1}\cdots\pi_{g_k,s_k}.\]
These projections in the definition commute, and so the resulting image by multiplying $\Cl(0,n)$ on the right with this is a subrepresentation.  The code for that subrepresentation will be $C$.

\section{Comparison to Representation theory of Clifford Algebras}
We have now seen that $\zt$-graded signed permutation representations of $\Cl(0,n)$ by signed permutation matrices are related to Cliffordinkras, which are characterized by doubly even binary linear block codes of length $n$.

At the same time, the representation theory of Clifford algebras is well-known.\cite{rLM}  In this section we illustrate the relationship with examples.

Before we begin, $\zt$-graded representations of $\Cl(0,n)$ are classified by ungraded representations of $\Cl(0,n+1)$.  That is because in a $\zt$-graded representation, there is an operator $(-1)^F$ that is $1$ on the bosons and $-1$ on the fermions.  This squares to $1$, and anticommutes with the $\Gamma_1,\ldots,\Gamma_n$, and so if we view $(-1)^F$ as $\Gamma_0$, then $\Gamma_0,\ldots,\Gamma_n$ furnish a representation of $\Cl(0,n+1)$.  Furthermore, $\zt$-graded isomorphisms of representations of $\Cl(0,n)$ are, in this framework, ungraded isomorphisms of representations of $\Cl(0,n+1)$.

As in Section~\ref{sec:cubes}, we have Cliffordinkras with $n$ colors that look like $n$-cubes.  This corresponds to the ``trivial'' code $t_n$, consisting only of the identity element: $t_n=\{00\cdots 0\}$.  In fact, for $n<4$, there are no other doubly even codes (a non-trivial codeword in a doubly even code would have to have weight at least 4, which is impossible if $n<4$).

For $n=1$, the construction in Section~\ref{sec:cubes} gives us a $1$-cube, that is, a line segment, with two vertices and one edge:
\begin{center}
\begin{tikzpicture}[scale=0.15]
\GraphInit[vstyle=Welsh]
\SetVertexNormal[MinSize=5pt]
\SetUpEdge[labelstyle={draw},style={ultra thick}]
\tikzset{Dash/.style={dashed,draw}}
\Vertex[x=0,y=0,Math,L={1},Lpos=180]{B}
\Vertex[x=20,y=0,Math,L={\Gamma_1},Lpos=0]{F}
\AddVertexColor{black}{F}
\Edge(B)(F)
\end{tikzpicture}
\end{center}
If we order the vertices so that $B$ comes first, this describes a two dimensional real representation, with
\[\Gamma_1=\begin{bmatrix}
0&1\\
1&0
\end{bmatrix}.\]
We also have
\[\Gamma_0=(-1)^F=\begin{bmatrix}
1&0\\
0&-1
\end{bmatrix}.\]
As mentioned above, $\zt$-graded representations of $\Cl(0,1)$ are ungraded representations of $\Cl(0,2)$, which according to Ref.~\cite{rLM}, for instance, is isomorphic to $\re(2)$, that is, $2\times 2$ real matrices.

The smallest real representation is the natural $2$-dimensional representation treating these matrices as themselves.  Thus we have agreement between the Cliffordinkra analysis and the classical representation theory.

For $n=2$, the Cliffordinkra is a square
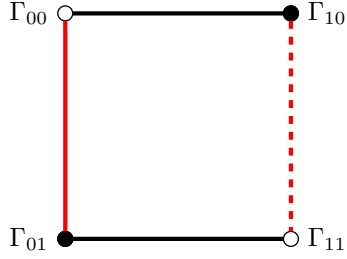
\begin{center}
\begin{tikzpicture}[scale=0.15]
\GraphInit[vstyle=Welsh]
\SetVertexNormal[MinSize=5pt]
\SetUpEdge[labelstyle={draw},style={ultra thick}]
\tikzset{Dash/.style={dashed,draw,ultra thick}}
\Vertex[x=0,y=0,Math,L={\Gamma_{00}},Lpos=180]{B}
\Vertex[x=20,y=0,Math,L={\Gamma_{10}},Lpos=0]{F}
\Vertex[x=0,y=-20,Math,L={\Gamma_{01}},Lpos=180]{C}
\Vertex[x=20,y=-20,Math,L={\Gamma_{11}},Lpos=0]{G}
\AddVertexColor{black}{F,C}
\Edge(B)(F)
\Edge(C)(G)
\Edge[color=red](B)(C)
\Edge[color=red,style=Dash](F)(G)
\end{tikzpicture}
\end{center}
which, depending on the ordering of the vertices, results in
\[\Gamma_1=\begin{bmatrix}
0&0&1&0\\
0&0&0&1\\
1&0&0&0\\
0&1&0&0\end{bmatrix},\quad\quad
\Gamma_1=\begin{bmatrix}
0&0&0&1\\
0&0&-1&0\\
0&1&0&0\\
-1&0&0&0\end{bmatrix}.\]
From the Clifford algebra classification standpoint, $\zt$-graded representations of $\Cl(0,2)$ are representations of $\Cl(0,3)\cong\bc(2)$.  The smallest real representation comes from the natural complex representation of $\bc(2)$ as complex matrices, then viewing a complex number as two real numbers (the real and imaginary parts).  This gives rise to a 4-dimensional representation, again, in agreement with what was seen with Cliffordinkras.

For $n=3$, the Cliffordinkra is a cube and thus gives an $8$-dimensional $\zt$-graded representation of $\Cl(0,3)$.  This is equivalent to a representation of $\Cl(0,4)\cong\bh(2)$.  If we take the standard representation of this as a quaternionic matrix, and treat every quaternion as four real numbers, then we get an $8$-dimensional real $\zt$-graded representation, in agreement with what Cliffordinkras provides.  The reader can use whatever ordering of the cube she prefers to get the $\Gamma_i$ matrices.

When $n=4$, we could take the trivial code $t_4$ again, but for the first time we have another option: the code $d_4=\{0000,1111\}$, generated by $1111$.  This is a code of dimension 1, described above in Section~\ref{sec:cliffordcube}, and the Cliffordinkra has $8$ vertices.  This thus provides a real $\zt$-graded representation of $\Cl(0,4)$, or a representation of $\Cl(0,5)\cong\bh(2)\oplus\bh(2)$.  There are two irreducible representations, given by each of these factors.  They again lead to $8$ dimensional real representations.  These correspond to the images of $\pi_{1111,+}$ and $\pi_{1111,-}$ given in Section~\ref{sec:cliffordcube}.

If we had taken the trivial code $t_4$ instead, the Cliffordinkra would be a $4$-cube, and the representation would not be minimal (it would be a direct sum of the two different irreducibles).  The relationship to the direct sum is not given by a signed permutation, but involves taking linear combinations of vertices.  This is essentially in Section~\ref{sec:cliffordcube} the way we broke $\Gamma_y$ into $a_y$ and $b_y$.

More generally, if we want irreducible representations of $\Cl(0,n)$ we should look for maximal doubly even codes, that is, of the highest dimension.

  When $n=5$, it is possible to have a codeword with four $1$s and one $0$, like $11110$, but it is impossible to have two such, because their bitwise and would have odd weight.  So we can have a generator set $11110$, or we could have the result of permuting the columns.

Given a binary block code, we can generate a new code by permuting the columns.  This procedure preserves linearity, evenness, and doubly evenness.  Two codes related by this procedure are called {\em permutation equivalent}.  So for $n=5$, there is only one doubly even linear code up to permutation equivalence: $\{00000,11110\}$.  Others, like $\{00000,01111\}$ are permutation equivalent to it.

Given two binary block codes $G$ and $H$, of length $m$ and $n$, respectively, we can generate the code $G\oplus H$, of length $m+n$, consisting of all strings resulting from concatening a codeword of $G$ with a codeword of $H$.  If $G$ and $H$ are linear, then so is $G\oplus H$.  If they are doubly even, then so is $G\oplus H$.  The code $\{00000,11110\}$ is $d_4\oplus t_1$.  And $t_n=t_1{}^{\oplus n}$.

So for $n=5$, we have $d_4\oplus t_1$ of dimension $k=1$, and the Cliffordinkra has $2^{n-k}=16$ vertices.  The $\zt$-graded representations of $\Cl(0,5)$ are representations of $\Cl(0,6)=\bh(4)$, and the natural representation of this, interpreting quaternions as four real numbers, gives a 16 dimensional real representation.

For $n=6$, the maximal doubly even code is called $d_6$.  It is dimension $k=2$, meaning we have two generators: $111100$ and $001111$.  It is common to record these generators as rows of a matrix, called the generating matrix:
\[d_6:\begin{bmatrix}
1&1&1&1&0&0\\
0&0&1&1&1&1
\end{bmatrix}\]

Since $k=2$, the quotient, the Cliffordinkra, has $2^{n-k}=16$ vertices.  The natural representation of $\Cl(0,7)=\bc(8)$ is using $8$ complex dimensions, and if we view a complex number as two real numbers, this gives $16$ real dimensions.

More generally, if $n$ is even and at least 4, we define $d_n$ to be the code of length $n$ and dimension $d/2-1$, with generating matrix
\[d_n:\begin{bmatrix}
1&1&1&1&0&0&\cdots&0&0&0\\
0&0&1&1&1&1&\cdots&0&0&0\\
0&0&0&0&1&1&\cdots&0&0&0\\
\vdots&&&&&&&&&\vdots\\
0&0&0&0&0&0&\cdots&1&0&0\\
0&0&0&0&0&0&\cdots&1&1&1
\end{bmatrix}.\]

For $n=7$ the maximal doubly even code has dimension $3$, and is given by the following generating matrix:
\[e_7:\begin{bmatrix}
1&1&1&1&0&0&0\\
0&0&1&1&1&1&0\\
1&0&1&0&1&0&1
\end{bmatrix}.\]

The Cliffordinkra here has $2^{n-k}=16$ vertices, and the natural representation of $\Cl(0,8)=\re(16)$ is 16 dimensional.

For $n=8$, the maximal code, $e_8$, has dimension $k=4$, and the following generating matrix:
\[e_8:\begin{bmatrix}
1&1&1&1&0&0&0&0\\
0&0&1&1&1&1&0&0\\
0&0&0&0&1&1&1&1\\
1&0&1&0&1&0&1&0
\end{bmatrix}.\]

The Cliffordinkra here has $2^{n-k}=16$ vertices, and there are two natural representations of $\Cl(0,9)=\re(16)\oplus \re(16)$, based on which factor is used.  Both can be obtained using projections as in Section~\ref{sec:cliffordcube}, but with different choices of $+$ or $-$ signs.  The details are below in Section~\ref{sec:cc}.

The naming of these codes comes from the fact that in $n$ dimensions, the set of lattice points $(x_1,\ldots,x_n)$ in $\zz^n$ with $(x_1\bmod{2},\ldots,x_n\bmod{2})$ being in $d_n$ (resp. $e_7$, $e_8$) is the lattice $d_n$ (resp. $e_7$, $e_8$), and the set of nonzero points that are closest to the origin forms the root lattice for the Lie groups $D_N$ (resp. $E_7$, $E_8$).

For these values of $n$, here is a table with the maximum code dimension, the number of vertices (which is the dimension of the representation from the Cliffordinkra), the $\Cl(0,n+1)$, and the minimal representation of this Clifford algebra:
\begin{center}
\begin{tabular}{c|c|c|c|c|c}
$n$&dim. code&max code&num. vert.&$\Cl(0,n+1)$&min. rep.\\\hline
1&0&$t_1$&2&$\re(2)$&$\re^2$\\
2&0&$t_2$&4&$\bc(2)$&$\bc^2$\\
3&0&$t_3$&8&$\bh(2)$&$\bh^2$\\
4&1&$d_4$&8&$\bh(2)\oplus\bh(2)$&$\bh^2$\\
5&1&$d_4\oplus t_1$&16&$\bh(4)$&$\bh^4$\\
6&2&$d_6$&16&$\bc(8)$&$\bc^8$\\
7&3&$e_7$&16&$\re(16)$&$\re^{16}$\\
8&4&$e_8$&16&$\re(16)\oplus\re(16)$&$\re^{16}$\\
\end{tabular}
\end{center}

For $n>8$, the best we can do is achieved by taking the direct sum of $e_8$ and whatever we had for length $n-8$.  For instance, the table continues:
\begin{center}
\begin{tabular}{c|c|c|c|c|c}
$n$&dim. code&max code&num. vert.&$\Cl(0,n+1)$&min. rep.\\\hline
9&4&$e_8\oplus t_1$&32&$\re(32)$&$\re^{32}$\\
10&4&$e_8\oplus t_2$&64&$\bc(32)$&$\bc^{32}$\\
11&4&$e_8\oplus t_3$&128&$\bh(32)$&$\bh^{32}$\\
12&5&$e_8\oplus d_4$&128&$\bh(32)\oplus\bh(32)$&$\bh^{32}$\\
13&5&$e_8\oplus d_4\oplus t_1$&256&$\bh(64)$&$\bh^{64}$\\
14&6&$e_8\oplus d_6$&256&$\bc(128)$&$\bc^{128}$\\
15&7&$e_8\oplus e_7$&256&$\re(256)$&$\re^{256}$\\
16&8&$e_8\oplus e_8$&256&$\re(256)\oplus\re(256)$&$\re^{256}$\\
\end{tabular}
\end{center}
This procedure continues for every block of $8$ values of $n$, in accordance with the familiar period-8 Bott Periodicity.\cite{rLM}

There are indeed other codes that are not of the form $e_8$ direct sum some smaller length code: for instance, $d_{10}$ is also dimension 4.  But these achieve the highest dimension possible.  We do not prove that here, but to get a sense for how this is done, here is a sketch.  First show that these are maximal in the sense that it is impossible to add another generator while maintaining doubly evenness.  Then show, using techniques analogous to matroid theory, that any two doubly even linear codes of a given length have the same dimension.

The classification of doubly even codes is a hard problem, because the number of codes grows quickly with $n$.

In 1996, Philippe Gaborit gave a count of all doubly even binary linear codes of length $n$ and dimension $k$, not up to permutation equivalence.\cite{rPGMass}  His formula is a bit long, and you can look it up there, but it is telling that it depends on $n\bmod{8}$.  This formula can also be used to calculate at what dimension $k$ the number of codes is $0$, and this matches the observation above.

The fact that the minimal Adinkra (obtained by a maximal code) for every $n$ gives the same dimension as that obtained from the irreducible representation of $\Cl(0,n+1)$ indicates that $\zt$-graded representations of $\Cl(0,n+1)$ can always be achieved by signed permutation matrices.

\section{Cubical Cohomology}
\label{sec:cc}
The study of the dashings on Cliffordinkras turns out to use the language of cubical cohomology from algebraic topology.

A Cliffordinkra, being a graph, contains vertices and edges.  Furthermore, there is a special role played by the quadrilateral bicolor cycles, because these are what define whether a dashing is totally odd.  We might imagine $2$ dimensional squares being glued for each quadrilateral bicolor cycle.  An analogous idea gives rise to $3$-cells.  The result is a cubical cell complex.

A $k$-cochain is a function from the set of $k$-cells to $\zt$.  The set of $k$-cochains is written $C^k(A)$.  This is a vector space over the field $\gf_2$.  There is a linear map $d_k:C^k(A)\to C^{k+1}(A)$ defined as follows: if $\mu\in C^k(A)$ and $\sigma$ is a $k+1$-cell, then  $d_k\mu(\sigma)=\mu(\partial \sigma)$, where $\partial\sigma$ is the formal sum of the cells in the boundary of $\sigma$.  As in any cell complex, $d_{k+1}\circ d_k=0$.  

A dashing is a specification of which edges are solid and which are dashed.  If we view solid as $0\in\zt$ and dashed as $1\in\zt$, then a dashing is an assignment of $0$ or $1$ to every edge.  Thus, a dashing is really a $1$-cochain in $C^1(A)$.  The condition that a dashing $\mu$ is totally odd is that for every $2$-cell $\sigma$, we have $\mu(\partial\sigma)=1\pmod{2}$.  This means $d_1\mu$ must be the $2$-cocycle $1_2\in C^2(A)$ whose value on every $2$-cell is $1$.

The study of the solutions to the equation $d_1\mu=1_2$ involves the first and second cohomology groups.  In general, the $k$th cohomology group, written $H^k(A;\zt)$, which is defined to be the kernel of $d_k:C^k(A)\to C^{k+1}(A)$ modulo the image of $d_{k-1}:C^{k-1}(A)\to C^k(A)$.  This is given by topological information in the cubical cell complex.

Whether the equation $d_1\mu=1_2$ has any solutions at all first requires that $d_2 1_2=0$.  For any $3$-cell $\sigma$, $d_21_2(\sigma)=0$ because every $3$-cell is a cube and has six faces, and six is even.  Whether or not $d_1\mu=1_2$ has any solutions then depends on viewing $1_2$ in $H^2(A;\zt)$ and determining if it is $0$.  For connected $A$, this turns out to be the case precisely when the code for $A$ is doubly even.  This is an independent verification of the fact that totally odd dashings occur precisely when the code is doubly even.

If totally odd dashings exist, we might ask how many there are.  If $\mu_1$ and $\mu_2$ are two totally odd dashings, then $d_1\mu_1=d_1\mu_2$, so $\mu_1-\mu_2$ is in the kernel of $d_1$.  In other words, the set of totally odd dashings is an affine space in $C^1(A)$ modeled on the kernel of $d_1$.

We now analyze vertex switching, which was introduced in Section~\ref{sec:vertexswitch}.  Let $S$ be a set of vertices.  This defines an indicator function $1_S(v)\in C^0(A)$ that is defined to be $1$ when $v\in S$ and $0$ otherwise.  If $\mu\in C^1(A)$ is a dashing, then $\mu+d_0(1_S)$ is the result of modifying $\mu$ by a vertex switch on $S$.  Therefore, the set of totally odd dashings up to vertex switching is an affine space modeled on the kernel of $d_1$ modulo the image of $d_0$; i.e., the first cohomology group $H^1(A;\zt)$.

This cohomology group $H^1(A;\zt)$ can be computed for connected Cliffordinkras, and it is isomorphic to the code of the Cliffordinkra.  Indeed, this is related to the $\pm$ choices in $\pi_{(g_1,\ldots,g_k),(\pm,\ldots,\pm)}$ in Section~\ref{sec:cliffordcube}.  This can be computed combinatorially but the topological intuition is this: the cube is contractible but when we quotient by a code $C$ we have $C$ as our fundamental group.  By the Hurewicz Isomorphism Theorem and the Universal Coefficient Theorem, we can compute $H^1(A;\zt)$.  The quotienting process can have fixed points but these are in cells of dimension 4 and higher, so that does not affect $H^1$.

%There are other isomorphisms between , however, due to the permutation of the vertices.  Ultimately, when everything is analyzed, there is up to isomorphism only one isomorphism class of representation for every maximal code, unless $n$ is a multiple of $4$, in which case there are two.  This corresponds to the fact that for these values of $n$, $\Cl(0,n+1)$ is a direct sum of two matrix algebras.

There is an analogous phenomenon with the bipartite structure: let $1_1$ be the $1$-cocycle that assigns $1$ to every $1$-cell.  The assignment of bosons and fermions is a $0$-cocycle $f$, so that $d_0f=1_1$.  This exists if $1_1$ is $0$ in $H^1(A;\zt)$, which in turn is the case if the code for each component of $A$ is even.  The ways to do this is an affine space modeled on $H^0(A;\zt)$, a $\zt$ vector space whose dimension is the number of components of $A$.

The details of all of this is done in Ref.~\cite{cc}.

\section{Geometrization}
\label{sec:geom}
The situation in the previous section, of the existence of a totally odd dashing depending on whether a certain $2$-cocycle is $0$ in $H^2(A;\zt)$, and the set of such, modulo sign reversals, being an affine space modeled on $H^1(A;\zt)$, is reminiscent of a similar story in differential topology.  A manifold $M$ has a spin structure if its second Stiefel--Whitney class $w_2$ is $0$ in $H^2(M;\zt)$, and the set of spin structures is given by an affine space modeled on $H^1(M;\zt)$.\cite{rMS}  It turns out that there is a sense in which totally odd dashings give rise to spin structures on surfaces, explained in Refs.~\cite{geom1,geom2}, which we summarize here.  This procedure is called {\em geometrization}.

A {\em rainbow} is a cyclic ordering of the colors in the Cliffordinkra.  Given a rainbow, we can glue in squares only for those bicolor cycles that involve consecutive colors in the rainbow.  If this is done, the result is a surface.  For instance, for an $n$-cube modulo a code of dimension $k$, we get $2^{n-k}$ vertices, $n2^{n-k-1}$ edges, and $n2^{n-k-2}$ two-cells.  An Euler characteristic calculation shows that we get that the surface has genus $1+(n-4)2^{n-k-3}$.

When $n=3$, this predicts a sphere, which corresponds to the $3$-dimensional cube that is the Cliffordinkra for $n=3$.  When $n=4$, no matter if we take $k=0$ for the code $t_4$ or $k=1$ for the code $d_4$, we have genus 1.  This can be seen by drawing the a projection of the hypercube and following the construction, as seen in Figure~\ref{fig:torus}.

\newcommand{\scr}{.7}
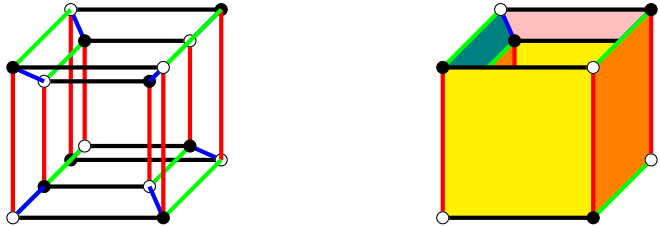
\begin{figure}[h]
\begin{center}
\begin{tikzpicture}[ultra thick]
\draw[color=red] (xyz cs:x=-1,y=-1,z=-1)--(xyz cs:x=-1,y=1,z=-1);
\draw[color=black] (xyz cs:x=-1,y=-1,z=-1)--(xyz cs:x=1,y=-1,z=-1);
\draw[color=black] (xyz cs:x=-1,y=1,z=-1)--(xyz cs:x=1,y=1,z=-1);
\draw[color=blue] (xyz cs:x=-1,y=-1,z=-1)--(xyz cs:x=-\scr,y=-\scr,z=-\scr);
\draw[color=blue] (xyz cs:x=1,y=-1,z=-1)--(xyz cs:x=\scr,y=-\scr,z=-\scr);
\draw[color=blue] (xyz cs:x=-1,y=1,z=-1)--(xyz cs:x=-\scr,y=\scr,z=-\scr);
\draw[color=blue] (xyz cs:x=1,y=1,z=-1)--(xyz cs:x=\scr,y=\scr,z=-\scr);
\draw[color=black] (xyz cs:x=-\scr,y=-\scr,z=-\scr)--(xyz cs:x=\scr,y=-\scr,z=-\scr);
\draw[color=red] (xyz cs:x=-\scr,y=-\scr,z=-\scr)--(xyz cs:x=-\scr,y=\scr,z=-\scr);
\draw[color=black] (xyz cs:x=-\scr,y=\scr,z=-\scr)--(xyz cs:x=\scr,y=\scr,z=-\scr);
\draw[color=red] (xyz cs:x=\scr,y=-\scr,z=-\scr)--(xyz cs:x=\scr,y=\scr,z=-\scr);
\filldraw[white,thin] (xyz cs:x=1,y=-1,z=-1) circle (.08);
\draw[black,thin] (xyz cs:x=1,y=-1,z=-1) circle (.08);
\filldraw[black,thin] (xyz cs:x=-1,y=-1,z=-1) circle (.08);
\filldraw[white,thin] (xyz cs:x=-1,y=1,z=-1) circle (.08);
\draw[black,thin] (xyz cs:x=-1,y=1,z=-1) circle (.08);
\filldraw[black,thin] (xyz cs:x=1,y=1,z=-1) circle (.08);
\draw[color=green] (xyz cs:x=-\scr,y=-\scr,z=-\scr)--(xyz cs:x=-\scr,y=-\scr,z=\scr);
\draw[color=green] (xyz cs:x=\scr,y=-\scr,z=-\scr)--(xyz cs:x=\scr,y=-\scr,z=\scr);
\draw[color=green] (xyz cs:x=-\scr,y=\scr,z=-\scr)--(xyz cs:x=-\scr,y=\scr,z=\scr);
\draw[color=green] (xyz cs:x=\scr,y=\scr,z=-\scr)--(xyz cs:x=\scr,y=\scr,z=\scr);
\filldraw[white,thin] (xyz cs:x=-\scr,y=-\scr,z=-\scr) circle (.08);
\draw[black,thin] (xyz cs:x=-\scr,y=-\scr,z=-\scr) circle (.08);
\filldraw[white,thin] (xyz cs:x=\scr,y=\scr,z=-\scr) circle (.08);
\draw[black,thin] (xyz cs:x=\scr,y=\scr,z=-\scr) circle (.08);
\filldraw[black,thin] (xyz cs:x=-\scr,y=\scr,z=-\scr) circle (.08);
\filldraw[black,thin] (xyz cs:x=\scr,y=-\scr,z=-\scr) circle (.08);
\draw[color=red] (xyz cs:x=1,y=-1,z=-1)--(xyz cs:x=1,y=1,z=-1);
\draw[color=green] (xyz cs:x=-1,y=-1,z=-1)--(xyz cs:x=-1,y=-1,z=1);
\draw[color=green] (xyz cs:x=1,y=-1,z=-1)--(xyz cs:x=1,y=-1,z=1);
\draw[color=green] (xyz cs:x=-1,y=1,z=-1)--(xyz cs:x=-1,y=1,z=1);
\draw[color=green] (xyz cs:x=1,y=1,z=-1)--(xyz cs:x=1,y=1,z=1);
\draw[color=black] (xyz cs:x=-\scr,y=-\scr,z=\scr)--(xyz cs:x=\scr,y=-\scr,z=\scr);
\draw[color=red] (xyz cs:x=-\scr,y=-\scr,z=\scr)--(xyz cs:x=-\scr,y=\scr,z=\scr);
\draw[color=black] (xyz cs:x=-\scr,y=\scr,z=\scr)--(xyz cs:x=\scr,y=\scr,z=\scr);
\draw[color=red] (xyz cs:x=\scr,y=-\scr,z=\scr)--(xyz cs:x=\scr,y=\scr,z=\scr);
\filldraw[white,thin] (xyz cs:x=\scr,y=-\scr,z=\scr) circle (.08);
\draw[black,thin] (xyz cs:x=\scr,y=-\scr,z=\scr) circle (.08);
\filldraw[white,thin] (xyz cs:x=-\scr,y=\scr,z=\scr) circle (.08);
\draw[black,thin] (xyz cs:x=-\scr,y=\scr,z=\scr) circle (.08);
\filldraw[black,thin] (xyz cs:x=\scr,y=\scr,z=\scr) circle (.08);
\filldraw[black,thin] (xyz cs:x=-\scr,y=-\scr,z=\scr) circle (.08);
\draw[color=blue] (xyz cs:x=-1,y=-1,z=1)--(xyz cs:x=-\scr,y=-\scr,z=\scr);
\draw[color=blue] (xyz cs:x=1,y=-1,z=1)--(xyz cs:x=\scr,y=-\scr,z=\scr);
\draw[color=blue] (xyz cs:x=-1,y=1,z=1)--(xyz cs:x=-\scr,y=\scr,z=\scr);
\draw[color=blue] (xyz cs:x=1,y=1,z=1)--(xyz cs:x=\scr,y=\scr,z=\scr);
\draw[color=black] (xyz cs:x=-1,y=-1,z=1)--(xyz cs:x=1,y=-1,z=1);
\draw[color=red] (xyz cs:x=-1,y=-1,z=1)--(xyz cs:x=-1,y=1,z=1);
\draw[color=black] (xyz cs:x=-1,y=1,z=1)--(xyz cs:x=1,y=1,z=1);
\draw[color=red] (xyz cs:x=1,y=-1,z=1)--(xyz cs:x=1,y=1,z=1);
\filldraw[white,thin] (xyz cs:x=-1,y=-1,z=1) circle (.08);
\draw[black,thin] (xyz cs:x=-1,y=-1,z=1) circle (.08);
\filldraw[black,thin] (xyz cs:x=1,y=-1,z=1) circle (.08);
\filldraw[white,thin] (xyz cs:x=1,y=1,z=1) circle (.08);
\draw[black,thin] (xyz cs:x=1,y=1,z=1) circle (.08);
\filldraw[black,thin] (xyz cs:x=-1,y=1,z=1) circle (.08);
\end{tikzpicture}
\makebox[1in]{}
\begin{tikzpicture}
\filldraw[color=pink] (xyz cs:x=-1,y=1,z=-1)--(xyz cs:x=1,y=1,z=-1)--(xyz cs:x=\scr,y=\scr,z=-\scr)--(xyz cs:x=-\scr,y=\scr,z=-\scr)--(xyz cs:x=-1,y=1,z=-1);
%\filldraw[color=lightgray] (xyz cs:x=-1,y=1,z=-1)--(xyz cs:x=1,y=1,z=-1)--(xyz cs:x=\scr,y=\scr,z=-\scr)--(xyz cs:x=-\scr,y=\scr,z=-\scr)--(xyz cs:x=-1,y=1,z=-1);
\filldraw[color=yellow] (xyz cs:x=-\scr,y=\scr,z=-\scr)--(xyz cs:x=\scr,y=\scr,z=-\scr)--(xyz cs:x=\scr,y=-\scr,z=-\scr)--(xyz cs:x=-\scr,y=-\scr,z=-\scr)--(xyz cs:x=-\scr,y=\scr,z=-\scr);
%\filldraw[color=lightgray] (xyz cs:x=-\scr,y=\scr,z=-\scr)--(xyz cs:x=\scr,y=\scr,z=-\scr)--(xyz cs:x=\scr,y=-\scr,z=-\scr)--(xyz cs:x=-\scr,y=-\scr,z=-\scr)--(xyz cs:x=-\scr,y=\scr,z=-\scr);
%
\draw[color=red,ultra thick] (xyz cs:x=-1,y=-1,z=-1)--(xyz cs:x=-1,y=1,z=-1);
\filldraw[color=teal] (xyz cs:x=-1,y=1,z=-1)--(xyz cs:x=-1,y=1,z=1)--(xyz cs:x=-\scr,y=\scr,z=\scr)--(xyz cs:x=-\scr,y=\scr,z=-\scr)--(xyz cs:x=-1,y=1,z=-1);
%\filldraw[color=lightgray] (xyz cs:x=-1,y=1,z=-1)--(xyz cs:x=-1,y=1,z=1)--(xyz cs:x=-\scr,y=\scr,z=\scr)--(xyz cs:x=-\scr,y=\scr,z=-\scr)--(xyz cs:x=-1,y=1,z=-1);
\filldraw[color=orange] (xyz cs:x=-\scr,y=\scr,z=-\scr)--(xyz cs:x=-\scr,y=\scr,z=\scr)--(xyz cs:x=-\scr,y=-\scr,z=\scr)--(xyz cs:x=-\scr,y=-\scr,z=-\scr)--(xyz cs:x=-\scr,y=\scr,z=-\scr);
%\filldraw[color=lightgray] (xyz cs:x=-\scr,y=\scr,z=-\scr)--(xyz cs:x=-\scr,y=\scr,z=\scr)--(xyz cs:x=-\scr,y=-\scr,z=\scr)--(xyz cs:x=-\scr,y=-\scr,z=-\scr)--(xyz cs:x=-\scr,y=\scr,z=-\scr);
%
\draw[color=black,ultra thick] (xyz cs:x=-1,y=-1,z=-1)--(xyz cs:x=1,y=-1,z=-1);
\draw[color=black,ultra thick] (xyz cs:x=-1,y=1,z=-1)--(xyz cs:x=1,y=1,z=-1);
\draw[color=blue,ultra thick] (xyz cs:x=-1,y=-1,z=-1)--(xyz cs:x=-\scr,y=-\scr,z=-\scr);
\draw[color=blue,ultra thick] (xyz cs:x=1,y=-1,z=-1)--(xyz cs:x=\scr,y=-\scr,z=-\scr);
\draw[color=blue,ultra thick] (xyz cs:x=-1,y=1,z=-1)--(xyz cs:x=-\scr,y=\scr,z=-\scr);
\draw[color=blue,ultra thick] (xyz cs:x=1,y=1,z=-1)--(xyz cs:x=\scr,y=\scr,z=-\scr);
\draw[color=black,ultra thick] (xyz cs:x=-\scr,y=-\scr,z=-\scr)--(xyz cs:x=\scr,y=-\scr,z=-\scr);
\draw[color=red,ultra thick] (xyz cs:x=-\scr,y=-\scr,z=-\scr)--(xyz cs:x=-\scr,y=\scr,z=-\scr);
\draw[color=black,ultra thick] (xyz cs:x=-\scr,y=\scr,z=-\scr)--(xyz cs:x=\scr,y=\scr,z=-\scr);
\draw[color=red,ultra thick] (xyz cs:x=\scr,y=-\scr,z=-\scr)--(xyz cs:x=\scr,y=\scr,z=-\scr);
\draw[color=green,ultra thick] (xyz cs:x=-\scr,y=-\scr,z=-\scr)--(xyz cs:x=-\scr,y=-\scr,z=\scr);
\draw[color=green,ultra thick] (xyz cs:x=\scr,y=-\scr,z=-\scr)--(xyz cs:x=\scr,y=-\scr,z=\scr);
\draw[color=green,ultra thick] (xyz cs:x=-\scr,y=\scr,z=-\scr)--(xyz cs:x=-\scr,y=\scr,z=\scr);
\draw[color=green,ultra thick] (xyz cs:x=\scr,y=\scr,z=-\scr)--(xyz cs:x=\scr,y=\scr,z=\scr);
\filldraw[color=orange] (xyz cs:x=1,y=1,z=-1)--(xyz cs:x=1,y=1,z=1)--(xyz cs:x=1,y=-1,z=1)--(xyz cs:x=1,y=-1,z=-1)--(xyz cs:x=1,y=1,z=-1);
%\filldraw[color=lightgray] (xyz cs:x=1,y=1,z=-1)--(xyz cs:x=1,y=1,z=1)--(xyz cs:x=1,y=-1,z=1)--(xyz cs:x=1,y=-1,z=-1)--(xyz cs:x=1,y=1,z=-1);
%
\draw[color=red,ultra thick] (xyz cs:x=1,y=-1,z=-1)--(xyz cs:x=1,y=1,z=-1);
\draw[color=green,ultra thick] (xyz cs:x=-1,y=-1,z=-1)--(xyz cs:x=-1,y=-1,z=1);
\draw[color=green,ultra thick] (xyz cs:x=1,y=-1,z=-1)--(xyz cs:x=1,y=-1,z=1);
\draw[color=green,ultra thick] (xyz cs:x=-1,y=1,z=-1)--(xyz cs:x=-1,y=1,z=1);
\draw[color=green,ultra thick] (xyz cs:x=1,y=1,z=-1)--(xyz cs:x=1,y=1,z=1);
\draw[color=black,ultra thick] (xyz cs:x=-\scr,y=-\scr,z=\scr)--(xyz cs:x=\scr,y=-\scr,z=\scr);
\draw[color=red,ultra thick] (xyz cs:x=-\scr,y=-\scr,z=\scr)--(xyz cs:x=-\scr,y=\scr,z=\scr);
\draw[color=black,ultra thick] (xyz cs:x=-\scr,y=\scr,z=\scr)--(xyz cs:x=\scr,y=\scr,z=\scr);
\draw[color=red,ultra thick] (xyz cs:x=\scr,y=-\scr,z=\scr)--(xyz cs:x=\scr,y=\scr,z=\scr);
\draw[color=blue,ultra thick] (xyz cs:x=-1,y=-1,z=1)--(xyz cs:x=-\scr,y=-\scr,z=\scr);
\draw[color=blue,ultra thick] (xyz cs:x=1,y=-1,z=1)--(xyz cs:x=\scr,y=-\scr,z=\scr);
\draw[color=blue,ultra thick] (xyz cs:x=-1,y=1,z=1)--(xyz cs:x=-\scr,y=\scr,z=\scr);
\draw[color=blue,ultra thick] (xyz cs:x=1,y=1,z=1)--(xyz cs:x=\scr,y=\scr,z=\scr);
\filldraw[color=yellow] (xyz cs:x=-1,y=1,z=1)--(xyz cs:x=1,y=1,z=1)--(xyz cs:x=1,y=-1,z=1)--(xyz cs:x=-1,y=-1,z=1)--(xyz cs:x=-1,y=1,z=1);
%\filldraw[color=lightgray] (xyz cs:x=-1,y=1,z=1)--(xyz cs:x=1,y=1,z=1)--(xyz cs:x=1,y=-1,z=1)--(xyz cs:x=-1,y=-1,z=1)--(xyz cs:x=-1,y=1,z=1);
%
\draw[color=black,ultra thick] (xyz cs:x=-1,y=-1,z=1)--(xyz cs:x=1,y=-1,z=1);
\draw[color=red,ultra thick] (xyz cs:x=-1,y=-1,z=1)--(xyz cs:x=-1,y=1,z=1);
\draw[color=black,ultra thick] (xyz cs:x=-1,y=1,z=1)--(xyz cs:x=1,y=1,z=1);
\draw[color=red,ultra thick] (xyz cs:x=1,y=-1,z=1)--(xyz cs:x=1,y=1,z=1);
\filldraw[black,thin] (xyz cs:x=-\scr,y=\scr,z=-\scr) circle (.08);
\filldraw[white,thin] (xyz cs:x=-1,y=1,z=-1) circle (.08);
\draw[black,thin] (xyz cs:x=-1,y=1,z=-1) circle (.08);
\filldraw[black,thin] (xyz cs:x=1,y=1,z=-1) circle (.08);
\filldraw[white,thin] (xyz cs:x=1,y=-1,z=-1) circle (.08);
\draw[black,thin] (xyz cs:x=1,y=-1,z=-1) circle (.08);
\filldraw[white,thin] (xyz cs:x=-1,y=-1,z=1) circle (.08);
\draw[black,thin] (xyz cs:x=-1,y=-1,z=1) circle (.08);
\filldraw[black,thin] (xyz cs:x=1,y=-1,z=1) circle (.08);
\filldraw[white,thin] (xyz cs:x=1,y=1,z=1) circle (.08);
\draw[black,thin] (xyz cs:x=1,y=1,z=1) circle (.08);
\filldraw[black,thin] (xyz cs:x=-1,y=1,z=1) circle (.08);
\end{tikzpicture}
\end{center}
\caption{Left: the hypercube, an Adinkra with $n=4$, $k=0$.  Right: the torus obtained using the rainbow: black, red, green, blue.}
\label{fig:torus}
\end{figure}

A totally odd dashing on the edges can be modified into a Kasteleyn orientation, which directs the edges using arrows.  The work of Cimasoni and Reshetikhin show that  the proprety of being totally odd results in a spin structure on the surface.

It is also interesting that these surfaces can be viewed as Riemann surfaces, and have a description as an algebraic curve over the complex numbers.  These come with a projection to the Riemann sphere so that the bosons are all preimages of $0$, the fermions are all images of $1$, and the preimage of the unit interval $[0,1]$ in the positive real axis in the complex plane has as preimage the Adinkra.  Thus, the Adinkra is then a Dessin d'enfant.\cite{girondo}

\section{Other signatures of Clifford algebras}
The same idea can be used to analyze $\Cl(n,0)$, and indeed, $\Cl(p,q)$, but there is a bit more involved.  Since $\Gamma_i$ would square not to $I$ but to $-I$, if $\Gamma_i v=w$, then $\Gamma_i w=-v$.  So if there is an edge between $v$ and $w$, we need to decide whether $v$ or $w$ determines whether the edge is solid or dashed.  One way to determine this is to pick the boson, since precisely one of $v$ and $w$ is bosonic.  We declare that the edge of color $i$ from $v$ to $w$ to be solid if $v$ is bosonic and $\Gamma_i v=w$.

For $\Cl(n,0)$, Cliffordinkras are precisely the same.  For $\Cl(p,q)$, we need to identify which colors square to $1$ and which square to $-1$.  Then the Cliffordinkras are almost the same, except that the quadrilaterals involving both kinds of colors must have an even number of dashed edges, not odd.  Connected Cliffordinkras are still quotients of cubes by codes, but since there are two kinds of colors, codewords are split into two substrings: the first $p$ bits and the last $q$ bits.  There are then two weights for a codeword: the weight for the first $p$ bits, and the weight for the last $q$ bits.  The condition is then that these two weights must be equal modulo 4.

\section{Supersymmetry}
\label{sec:susy}
As indicated in the introduction, Adinkras first came about from the study of supersymmetry in subatomic particle physics.  Some particles in nature are bosons, and the others are fermions.\cite{rDF} Supersymmetry involves operators $Q_1,\ldots,Q_n$ that exchange bosons and fermions, and must satisfy relations of the type:
\[\{Q_j,Q_k\}=\sum_\mu C_{jk\mu}\frac{\partial}{\partial x^\mu}\]
for some constants $C_{jk\mu}$.  These constants satisfy certain consistency conditions.  Adinkras come about from considering the case of supersymmetry in one dimension, namely, time, where the relations become\cite{rDF}
\begin{equation}
\{Q_j,Q_k\}=2i\,\delta_{jk} \frac{\partial}{\partial t}.
\label{eqn:susy}
\end{equation}
This is the one dimensional supersymmetry algebra.

The motivation for studying supersymmetry in one dimension is not merely that it is an easy starting point, although it is that.  Supersymmetry in any number of dimensions can be studied from the one dimensional point of view, by ignoring the effects on the spatial directions.  Furthermore, some systems are naturally understood as depending only on time.

Since the 1970s, physicists have been studying representations of the supersymmetry algebra in a somewhat ad hoc manner, often using constructions like superspace.  They discovered that many of their constructions resulted in descriptions of the representations where $Q_i$ of one basis element only involved one other basis element.

Based on this observation, in 2004 physicists Michael Faux and Sylvester James Gates, Jr. invented Adinkras to study supersymmetry in one dimension.\cite{rA}  The idea is that if $B$ is a boson and $F$ is a fermion, we might have
\[Q_j B=F\]
in which case, in light of the supersymmetry algebra (\ref{eqn:susy}), we would also have
\[Q_j F=i\frac{d}{dt} B.\]
In this case we draw an arrow from $B$ to $F$:

\begin{center}
\begin{tikzpicture}[scale=0.15]
\GraphInit[vstyle=Welsh]
\SetVertexNormal[MinSize=5pt]
\SetUpEdge[labelstyle={draw},style={ultra thick,->}]
\tikzset{Dash/.style={dashed,draw}}
\Vertex[x=0,y=0,Math,L={B},Lpos=180]{B}
\Vertex[x=20,y=0,Math,L={F},Lpos=0]{F}
\AddVertexColor{black}{F}
\Edge(B)(F)
\end{tikzpicture}
\end{center}

But the derivative might very well happen on the $Q_j$ acting on the boson:
\[Q_j B=\frac{d}{dt}F\]
And then the supersymmetry algebra (\ref{eqn:susy}) requires
\[Q_j F=i B.\]
The Adinkra then is drawn in this way:
\begin{center}
\begin{tikzpicture}[scale=0.15]
\GraphInit[vstyle=Welsh]
\SetVertexNormal[MinSize=5pt]
\SetUpEdge[labelstyle={draw},style={ultra thick,->}]
\tikzset{Dash/.style={dashed,draw}}
\Vertex[x=0,y=0,Math,L={B},Lpos=180]{B}
\Vertex[x=20,y=0,Math,L={F},Lpos=0]{F}
\AddVertexColor{black}{F}
\Edge(F)(B)
\end{tikzpicture}
\end{center}

The attentive reader might notice that there is also the question of where to put the factor of $i$, but this plays no role in the representation theory because it can be absorbed into the redefinition of the $B$ or $F$; in any case the reality conditions from the physics demand that they be on the action of $Q_j$ on the fermions.

When there are at least two colors, we must have the same number of derivatives in $Q_jQ_k$ as $Q_kQ_j$, since they are negatives.  So there is an additional requirement that around any bicolor cycle, if traversed in some fixed direction, there must be as many arrows going along the path as against it.  For instance:

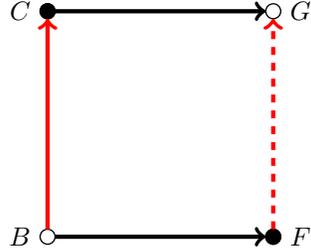
\begin{center}
\begin{tikzpicture}[scale=0.15]
\GraphInit[vstyle=Welsh]
\SetVertexNormal[MinSize=5pt]
\SetUpEdge[labelstyle={draw},style={ultra thick,->}]
\tikzset{Dash/.style={dashed,draw,ultra thick,->}}
\Vertex[x=0,y=0,Math,L={B},Lpos=180]{B}
\Vertex[x=20,y=0,Math,L={F},Lpos=0]{F}
\Vertex[x=0,y=20,Math,L={C},Lpos=180]{C}
\Vertex[x=20,y=20,Math,L={G},Lpos=0]{G}
\AddVertexColor{black}{F,C}
\Edge(B)(F)
\Edge(C)(G)
\Edge[color=red](B)(C)
\Edge[color=red,style=Dash](F)(G)
\end{tikzpicture}
\end{center}

If we start at $C$ and go clockwise, we go along the arrow to $G$, against the arrow to $F$, against the arrow to $B$, and with the arrow back to $C$.  Likewise if we start at any vertex and go in any fixed direction, there will be two arrows along the path and two against it.

More generally, an Adinkra is like a Cliffordinkra except that the edges are directed, satisfying the condition that that along any quadrilateral, one traverses two arrows going with the path and two arrows going against it.

This added data can be viewed as providing an integer value to each vertex, which can be viewed as the height on the page where the vertex is drawn, where the arrow points from a vertex of one height to a vertex of height one higher.  Such an assignment is called a {\em ranking}, and the set of rankings forms a poset, and relates to the number of three-colorings of the graph.  The study of these rankings is covered in Ref.~\cite{zhang}.

This ranking has physical meaning, in terms of the kinds of units used to describe the particles: an arrow always points from a particle with mass units $[M]^k$ to one with mass unit $[M]^{k+1/2}$.  This means that the ranking can be taken to be the exponent of the mass unit divided by two.

\section{Conclusion}
Diagrams like Cliffordinkras are well suited to understand Clifford algebras because the real representations turn out to be isomorphic to signed permutation representations.  This can be verified case by case up to dimension 8, and then Bott periodicity proves it for higher $n$, but an overall organizing principle responsible for this seems to be lacking at this point.

The relationship between the supersymmetry algebra in one dimension and Clifford algebras then explains why Adinkras are successful in describing representations of the supersymmetry algebra.  The placement of the derivatives need not be in line with the basis given by the signed permutation representation, but for many representations of interest to physicists, they are.  This in turn is because of the mass units, and because there is still enough freedom to choose bases so that finding a basis where the representation is a signed permutation representation and where the derivatives respect this is not uncommon in small examples.  An example where this does not happen is due to Gates, Hallet, H\"ubsch, and Stiffler.\cite{ComplexLinear}

Nevertheless when focused on Clifford algebras, Cliffordinkras allow a confluence of the study of doubly even codes, Clifford algebras, and discrete algebraic topology.

\section{Acknowledgments}
The author was partially supported by the endowment of the Ford Foundation Professorship of Physics at Brown University, and by the U.S. National Science Foundation grant PHY-1315155.

\section{References}
\bibliographystyle{alpha}
\bibliography{susyrefs}

\end{document}